\documentclass{llncs}
\usepackage[utf8]{inputenc}
\usepackage[T1]{fontenc}
\usepackage{lmodern}
\usepackage{microtype}
\usepackage[english]{babel}

\usepackage{mathtools}
\usepackage{amsmath,amssymb,amsfonts}
\usepackage{amsthm}
\usepackage[colorlinks,linkcolor={blue}]{hyperref}
\usepackage{cleveref}
\usepackage{enumitem}

\crefname{corollary}{Corollary}{Corollaries}
\Crefname{corollary}{Corollary}{Corollaries}
\crefname{lemma}{Lemma}{Lemmas}
\Crefname{lemma}{Lemma}{Lemmas}
\crefname{theorem}{Theorem}{Theorems}
\Crefname{theorem}{Theorem}{Theorems}
\crefname{observation}{Observation}{Observations}
\Crefname{observation}{Observation}{Observations}
\crefname{definition}{Definition}{Definitions}
\Crefname{definition}{Definition}{Definitions}
\crefname{proposition}{Proposition}{Propositions}
\Crefname{proposition}{Proposition}{Propositions}
\crefname{section}{Section}{Sections}
\Crefname{section}{Section}{Sections}
\crefname{figure}{Figure}{Figures}
\Crefname{figure}{Figure}{Figures}
\crefname{equation}{}{}
\Crefname{equation}{}{}

\renewcommand{\subset}{\subseteq}

\newcommand{\ceil}[1]{\left\lceil{#1}\right\rceil}
\newcommand{\floor}[1]{\left\lfloor{#1}\right\rfloor}
\newcommand{\cond}{\mathrel{}\middle\vert\mathrel{}}

\newcommand{\Oo}{\mathcal O} %big-O notation
\newcommand{\abs}[1]{\left \lvert #1 \right \rvert}
\newcommand{\set}[1]{\left \{ #1 \right \}}

\newcommand{\N}{{\mathbb{N}}}
\newcommand{\E}{\mathbb{E}}

\newcommand\drop[1]{}

\newcommand{\calF}{\mathcal{F}}
\newcommand{\calQ}{\mathcal{Q}}
\newcommand{\calA}{\mathcal{A}}
\newcommand{\etal}{{\em et al.}}
\usepackage[pdftex]{graphicx}
\usepackage[usenames,dvipsnames,table]{xcolor}
\definecolor{shade}{RGB}{235,235,235}
\usepackage{tikz}
\usetikzlibrary{calc}

\usepackage{environ}
\makeatletter
\newsavebox{\box@tikzpicture}
\NewEnviron{tikzpicture*}[2][]{%
  \begin{lrbox}{\box@tikzpicture}%
    \begin{tikzpicture}[#1]
      \BODY
    \end{tikzpicture}%
  \end{lrbox}
  \pgfmathsetmacro\width@scale@picture{#2/\wd\box@tikzpicture}%
  \begin{tikzpicture}[#1,scale=\width@scale@picture]
    \BODY
  \end{tikzpicture}%
}%
\makeatother

\usepackage{todonotes}

\newcommand*\samethanks[1][\value{footnote}]{\footnotemark[#1]}
\usepackage{authblk}

\title{Near-Optimal Induced Universal Graphs for Bounded Degree Graphs}
\author{Mikkel Abrahamsen \thanks{Research partly supported by Mikkel Thorup's
    Advanced Grant from the Danish Council for Independent Research
    under the Sapere Aude research career programme.} \and Stephen Alstrup \thanks{Research partly supported by the FNU
    project AlgoDisc -- Discrete Mathematics, Algorithms, and Data Structures.} \and Jacob Holm \and Mathias Bæk Tejs Knudsen \samethanks[1]\ \samethanks[2]
    \and Morten Stöckel\thanks{Research supported by Villum Fonden.}}

\institute{University of Copenhagen \\ \email{ \{miab,s.alstrup,jaho,knudsen,most\}@di.ku.dk}}

\date{}

\begin{document}
\setcounter{page}{0}
\maketitle

\begin{abstract}
%bemaerk at for odd value af D er der lavet forsimpling derfor er det ikke lige til at gennemskue udsagnet D>2 
A graph $U$ is an induced universal graph for a family $\mathcal{F}$ of graphs if every graph in $\mathcal{F}$ is a vertex-induced subgraph of $U$. For the family of all undirected graphs on $n$ vertices Alstrup, Kaplan, Thorup, and Zwick [STOC 2015] give an induced universal graph with $O\!\left(2^{n/2}\right)$ vertices, matching a lower bound by Moon [Proc. of Glasgow Math. Association 1965].

%Improving previous results by Butler and Esperet, and Arnaud and Ochem, we give an induced universal graph with $O\!\left((2n)^{\ceil{D/2}}\right)$ vertices for the family of graphs with bounded degree $D$. For $D$ being constant Butler also give a lower bound $\Omega\!\left(n^{D/2}\right)$. For an odd constant $D\geq 3$ other techniques can be used, and Esperet \etal~and Alon and Capalbo have shown a graph of $O\!\left(n^{\ceil{\frac{D}{2}}-\frac{1}{D}}\right)$ vertices. Using such techniques for even values of $D$ give asymptotically worse bounds than we present. We also present another upper bound, for any values of $D$,  

\newcommand{\changedmbtk}[1]{{\color{blue} #1}}
\newcommand{\removedmbtk}[1]{}

Let $k=\ceil{D/2}$. Improving asymptotically on previous results by Butler [Graphs and Combinatorics 2009] and Esperet, Arnaud and Ochem [IPL 2008], we give an induced universal graph with $O\!\left(\frac{k2^k}{k!}n^k \right)$ vertices for the family of graphs with $n$ vertices of maximum degree $D$. For constant $D$, Butler gives a lower bound of $\Omega\!\left(n^{D/2}\right)$.
\removedmbtk{Our upper bound is the first to beat this lower bound when $D$ is sufficiently large.}
For an odd constant $D\geq 3$, Esperet \etal~and Alon and Capalbo [SODA 2008] give a graph with $O\!\left(n^{k-\frac{1}{D}}\right)$ vertices. Using their techniques for any (including constant) even values of $D$ gives asymptotically worse bounds than we present.
\removedmbtk{For any $D$ we also give a lower bound $\Omega\!\left((\frac{n}{2eD})^{D/2}\right)$.}

For large $D$, i.e. when $D = \Omega\left(\log^3 n\right)$, the previous best upper bound was $\binom{n}{\ceil{D/2}} n^{O(1)}$ due to
Adjiashvili and Rotbart [ICALP 2014]. We give upper and lower bounds showing that the size is
$\binom{\floor{n/2}}{\floor{D/2}}2^{\pm\tilde{O}\left(\sqrt{D}\right)}$, where $\tilde{O}$ hides log-factors in $n$ and $D$.
Hence the optimal size is $2^{\tilde{O}(D)}$ and our construction comes within a factor of $2^{\tilde{O}\left(\sqrt{D}\right)}$ from
this. The previous results were larger by at least a factor of $2^{\Omega(D)}$.

\removedmbtk{
We give tighter lower and upper bounds of
\[
\binom{\floor{n/2}}{\floor{D/2}}
\cdot 2^{-O \!\left (\sqrt{D\log n} \cdot \log(n/D) \right )} \text{ and } \binom{\floor{n/2}}{\floor{D/2}} \cdot 2^{O \!\left (\sqrt{D\log n} \cdot \log(n/D) \right )},
\]
respectively, where the upper bound is a randomized construction.}

Part of our solution for the above results is to construct an induced universal graph with $2n-1$ vertices for the family of graphs with maximum degree $2$. 
This proves the correctness of a conjecture by Esperet \etal~stating the existence of such a graph with $2n+o(n)$ vertices.

In addition, we give results for acyclic graphs with maximum degree $2$, cycle graphs, and investigate the relationship between induced universal graphs and adjacency labeling schemes. From a labeling perspective our bounds are the first to give labels for any value of $D$ that are at most $o(n)$ bits longer than the shortest possible labels.

\end{abstract}

\thispagestyle{empty}
\newpage
\setcounter{page}{1}

\section{Introduction}
A graph $G=(V, E)$ is said to be an \emph{induced universal graph} for a family $\cal F$ of graphs if it contains each graph in $\cal F$ as a vertex-induced subgraph. A graph $H=(V',E')$ is contained in $G$ as a \emph{vertex-induced subgraph} if $V'
\subseteq V$ and $E'=\{vw\mid v,w \in V' \wedge vw \in E\}$. Induced universal graphs have been studied since the 1960s~\cite{moon1965minimal,Rado64}, and bounds on
the sizes of induced universal graphs have been given
for many families of graphs, including general, bipartite~\cite{AlstrupKTZ14}, and bounded arboricity graphs~\cite{adjacencytrees2015}.
We later define the classic distributed data structure \emph{adjacency labeling scheme} and describe how it is directly related to induced universal graphs. In Table~\ref{tab:adjacency2} in \Cref{sec:overview} below we give an overview of previous results and results in this paper.

\subsection{Overview of new and existing results}\label{sec:overview}
We give an overview in Table~\ref{tab:adjacency2} of dominating existing and new results. All bounds are on sizes of induced universal graphs. In the table, ``P'' refers to a result in this paper, $k=\ceil{D/2}$, $L=(\sqrt{D\log n} \cdot \log(n/D))$ and $U=(\sqrt{D\log n} \cdot \log(n/D))$. The ``A'' and ``B'' case below represent two different constructions. The upper bound in ``B'' is a randomized construction, whereas both lower bounds hold for both upper bounds.

\begin{table*}[ht]
	\renewcommand{\arraystretch}{1.3}
	\small
	\centering
	\makebox[0pt][c]{
		\begin{tabular}{|c|c|c|c|}
			%\toprule
			%\cline{2-5}
			\hline
			
			% \vspace*{3pt}
			\hline
			\bf Graph family & \bf Lower bound & \bf Upper bound & \bf Lower/Upper\\
			\hline
			\noalign{\vskip 2mm} \hline
			%\midrule
			General&   $2^{\frac{n-1}{2}}$ & $O( 2^{\frac{n}{2}})$ &  \cite{moon1965minimal}/ \cite{AlstrupKTZ14} \\
			\hline
			Tournaments &   $2^{\frac{n-1}{2}}$ & $O( 2^{\frac{n}{2}})$ &  \cite{moon1968topics}/\cite{AlstrupKTZ14}  \\
			\hline
			Bipartite&   $\Omega(2^{\frac{n}{4}})$ & $O( 2^{\frac{n}{4}})$ & \cite{Lozin2007}/\cite{AlstrupKTZ14} % (\cite{alstruprauhe})
			\\
			\hline \noalign{\vskip 2mm}
			\hline
			A: Max degree $D$ & $\Omega((\frac{n}{2eD})^{D/2})$ & $O\!\left(\frac{\min (n,k2^k)}{k!}n^k \right)$ &  P/P and ~\cite{icalpnoy14} \\
			\hline
			B: Max degree $D$ & $\binom{\floor{n/2}}{\floor{D/2}}
			\cdot 2^{-O(L)}$ & $ \binom{\floor{n/2}}{\floor{D/2}} \cdot 2^{O(U)}$ &  P/P \\
			\hline
			Max degree $D = o\!\left(\sqrt{n}\right)$ or $D \ge \frac{(2/3+\Omega(1))n}{\ln n}$ &
			$\binom{\floor{n/2}}{\floor{D/2}}\cdot n^{-O(1)}$ & 
			&  \cite{mckay1990asymptotic,mckay1991asymptotic} \\
			\hline
			Constant odd degree $D$ & $\Omega(n^{\frac{D}{2}})$ &  $O(n^{k-\frac{1}{D}})$ & \cite{Butler_induced-universalgraphs}/\cite{privatealon,AlonCapalbo2008,Esperet2008} \\
			\hline
			Max degree 2&  $11 \floor{n/6}$ & $2n-1$& \cite{Esperet2008}/P    \\
			\hline
			Acyclic, max degree 2&  $\floor{3/2n}$ & $\floor{3/2n}$& P/P    \\ 
			\hline
			A cycle aware of $n$&  $n+ \Omega(\log \log n)$ & $n+\log n + O(1)$  & P/P    \\
			\hline
			A cycle not aware of $n$&  $n + \Omega\!\left(\sqrt[3]{n}\right)$ & $n+ O(\sqrt{n})$ & P/P    \\
			\hline
			\noalign{\vskip 2mm} \hline
			%Families with an excluded minor
			Excluding a fixed minor
			& $\Omega(n)$ & $n^2 (\log {n})^{O(1)} $ & \cite{gavoille2007shorter}  \\
			\hline
			Planar& $\Omega(n)$ &$n^2(\log n)^{O(1)}$ &  \cite{gavoille2007shorter} \\
			\hline
			Planar, constant degree & $\Omega(n)$ & $O(n^2)$& \cite{Chung90} \\
			\hline
			Outerplanar & $\Omega(n)$ &  $n (\log n)^{O(1)}$ &  \cite{gavoille2007shorter} \\
			\hline
			Outerplanar, constant degree & $\Omega(n)$& $O(n)$& \cite{Chung90}\\
			\hline
			\noalign{\vskip 2mm} \hline
			Treewidth $l$ & $n2^{\Omega(l)}$& $n (\log \frac{n}{l})^{O(l)} $& \cite{gavoille2007shorter}\\
			
			\hline
			Constant arboricity $l$ & $\Omega(n^l)$ & $O(n^l)$ & \cite{alstruprauhe}/\cite{adjacencytrees2015} \\
			\hline
		\end{tabular}
	}
	%\caption{Induced-universal graphs.Above bounded and fixed means the degree assumed being a constant}.
	\caption{Induced-universal graphs for various families of graphs.  ``P'' is results in this paper. For the max degree results $k=\ceil{D/2}$. In the result for families of graphs with an excluded minor, the~$O(1)$ term in the exponent depends on the fixed minor excluded.}
	\label[table]{tab:adjacency2}
\end{table*}

\subsection{Maximum degree $2$ and maximum degree $D$}
Let $g_v(\calF)$ be the smallest number of vertices in any induced universal graph for a family of graphs $\calF$. Let $\mathcal{G}_D$ be the family of graphs with $n$ vertices and maximum degree $D$. In the families of graphs we study in this paper, a graph always has $n$ vertices, unless explicitly stated otherwise.

{\bf {\large Maximum degree $2$.}} Butler~\cite{Butler_induced-universalgraphs} shows that $g_v(\mathcal{G}_2)\leq 6.5n$. That was subsequently improved to $g_v(\mathcal{G}_2)\leq 2.5n+O(1)$ by Esperet \etal~\cite{Esperet2008} who also show the lower bound $g_v(\mathcal{G}_2)\geq 11 \floor{n/6}$. Esperet \etal~\cite{Esperet2008} conjecture that $g_v(\mathcal{G}_2)\leq 2n+o(n)$
and raise as an open problem to prove or disprove this.
We show the correctness of the conjecture by proving $g_v(\mathcal{G}_2)\leq 2n-1$. The $11 \floor{n/6}$ lower bound is based on a family of graphs whose largest component has $3$ vertices. We show matching $\frac{11}{6}n+\Oo(1)$ upper bounds for the family of graphs in $\mathcal{G}_2$ whose largest component is sufficiently small ($\leq6$ vertices), and for the family of graphs in $\mathcal{G}_2$ whose smallest component is sufficiently large ($\geq10$ vertices).

{\bf {\large Maximum degree $D$.}}
Let $k=\ceil{D/2}$. To give an upper bound for any value of $D$, Butler~\cite{Butler_induced-universalgraphs} first establishes: 

\begin{corollary}[\cite{Butler_induced-universalgraphs}]\label[corollary]{Butlersplit}
	Let $G \in \mathcal{G}_D$ be a graph on $n$ vertices with maximum degree $D$. Then $G$ can be decomposed into $k$ edge disjoint subgraphs where the maximum degree of each subgraph is at most $2$.
\end{corollary}

To achieve an upper bound this can be combined with:

\begin{theorem}[\cite{Chung90}] \label[theorem]{ChungSplit}
	Let $\calF$ and $\calQ$ be two families of graphs and let $G$ be an
	induced universal graph for $\calF$.
	Suppose that every graph in the family $\mathcal{Q}$ can be edge-partitioned
	into $\ell$ parts, each of which forms a graph in $\mathcal{F}$. Then $g_v(\mathcal{Q}) \leq |V[G]|^\ell$.
\end{theorem}

Butler~\cite{Butler_induced-universalgraphs} concludes $g_v(\mathcal{G}_D)\leq (6.5n)^k $. Similarly Esperet \etal~\cite{Esperet2008} achieve $g_v(\mathcal{G}_D)\leq (2.5n+O(1))^k$, and we achieve $g_v(\mathcal{G}_D)\leq (2n-1)^k=O(2^k n^k)$.

For constant maximum degree $D$, Butler~\cite{Butler_induced-universalgraphs} also shows $g_v(\mathcal{G}_D)=\Omega(n^{D/2})$. 

When $D$ is even and constant, the bounds are hence very tight:
$g_v(\mathcal{G}_D)=\Theta(n^{D/2})$. However, for non-constant $D$ we can, using another approach but still building on top of our maximum degree $2$ solution, beat Butler's lower bound for constant degree: For any value of $D$, we prove the upper bound $g_v(\mathcal{G}_D)=O\!\left(\frac{k2^k}{k!}n^k \right)$.
We also give a lower bound for any value of $D$: $\Omega\!\left((\frac{n}{2eD})^{D/2}\right)$.

{\bf {\large Constant odd degree.}}
A \emph{universal} graph for a family of graphs ${\cal F}$ is a graph that contains each graph from ${\cal F}$ as a subgraph (not necessarily vertex induced). The challenge is to construct universal graphs with as few edges as possible.
%Let $f_e(\calF)$ denote the minimum number of edges in a universal graph for $\calF$. 
%Mikkel: this notation is never used in the introduction.

A graph has \emph{arboricity} $k$ if the edges of the graph can be partitioned into at most $k$ forests. Graphs with maximum degree $D$ have arboricity bounded by $\floor{\frac{D}{2}}+1$~\cite{chartrand68,Lovasz66}.

When $D$ is odd and constant, some improvement have been achieved \cite{AlonCapalbo2008,Esperet2008} on the above bounds on $g_v(\mathcal{G}_D)$
by arguments involving
universal graphs and graphs with bounded arboricity. Let $\mathcal{A}_k$ denote a family of graphs with arboricity at most $k$.
\begin{theorem}[\cite{Chung90}] \label[theorem]{Arboricity}
 Let $G$ be a universal graph for $\mathcal{A}_k$ and $d_i$ the degree of vertex $i$ in $G$. Then $g_v(\mathcal{A}_k) \leq \sum_{i}(d_i+1)^k$.
\end{theorem}

Alon and Capalbo~\cite{alon2007sparse} describes a universal graph with $n$ vertices of maximum degree $c(D)n^{1-2/D}\log^{4/D}n$ for the family $\mathcal{G}_D$,
where $D\geq 3$ and $c(D)$ is a constant. Using this bound in Theorem~\ref{Arboricity},
Esperet \etal~\cite{Esperet2008} note that for odd $D$ (and hence arboricity $k=\ceil{\frac{D}{2}}$), we get
$g_v(\mathcal{G}_D)\leq c_1(D)n^{k-\frac{1}{D}}\log^{2+\frac{2}{D}}n$, for a constant $c_1(D)$.\footnote{In~\cite{Esperet2008} a typo states that the maximum degree for the universal graph in~\cite{alon2007sparse} is $c(D)n^{2-2/D}\log^{4/D}n$. The theorem in~\cite{alon2007sparse} only states the total number of edges being $c(D)n^{2-2/D}\log^{4/D}n$, however the maximum degree is $c(D)n^{1-2/D}\log^{4/D}n$~\cite{privatealon}.}
Using the slightly better universal graphs from~\cite{AlonCapalbo2008} the maximum degree is reduced to $c(D)n^{1-2/D}$~\cite{privatealon}, giving $g_v(\mathcal{G}_D)\leq c_2(D)n^{k-\frac{1}{D}}$, for a constant $c_2(D)$. Note that using this technique for even values of $D$ would give $g_v(\mathcal{G}_D)\leq c_3(D)n^{\frac{D}{2}+1-\frac{2}{D}}$, for a constant $c_3(D)$, which is asymptotically worse even for constant values of $D$, compared to any of the new upper bounds presented in this paper.
 
In~\cite{alon2010universality} it is stated that the methods in~\cite{AlonCapalbo2008} can be used to achieve $g_v(\mathcal{G}_D)=O(n^{D/2})$ for constant odd values of $D>1$, however according to~\cite{privatealon} this still has to be checked more carefully, and the hidden constant in the $O$-notation is not small. 
%More specifically (and with same comments~\cite{privatealon}) it is mentioned in ~\cite{privatealon,Butler_induced-universalgraphs} that techniques from~\cite{AlonCapalbo2008} will give $g_v(G_3)=O(n^{3/2})$. 

\subsection{Adjacency labeling schemes and induced universal graphs}
An \emph{adjacency labeling scheme} for a given family $\calF$ of graphs assigns 
\emph{labels} to the vertices of each graph in $\calF$ such that a \emph{decoder} given the
labels of two vertices from a graph, and no other information, can determine whether or not the vertices are adjacent in the graph. The labels
are assumed to be bit strings, and the goal is to minimize the maximum label
size. A $b$-bit labeling scheme uses at most $b$ bits per label. Information theoretical studies of adjacency
labeling schemes go back to the 1960s~\cite{Breuer66,BF67}, and efficient labeling schemes were introduced in~\cite{KNR92,muller}.
For graphs with bounded degree $D$, it was shown in~\cite{BF67} that labels of size $2nD$ can be constructed such that two vertices are adjacent whenever the Hamming distance~\cite{hamming} of their labels is at most $4D-4$.
A labelling scheme for $\calF$ is said to have \emph{unique labels} if no two vertices in the same graph from
$\calF$ are given the same label.

\begin{theorem}[\cite{KNR92}] \label[theorem]{KNRreduction}
	A family $\calF$ of graphs has a $b$-bit adjacency labeling scheme with
	unique labels iff $g_v(\calF) \leq 2^b$.
\end{theorem}

From a labeling perspective the above new upper and lower bounds are at most an additive $O(D+\log n)$ term from optimality.

\subsection{Better bounds for larger $D$, $D=\Omega(\log^3 n)$}
We have another approach which for large $D$, $D=\Omega(\log^3 n)$, gives better bounds than the ones presented above for constant $D$. The previous best upper bound for such large $D$ was $\binom{n}{\ceil{D/2}} n^{O(1)}$ due to Adjiashvili and Rotbart~\cite{icalpnoy14}. For any $D$ we prove the lower and upper bounds
\[\binom{\floor{n/2}}{\floor{D/2}}
\cdot 2^{-O \!\left (\sqrt{D\log n} \cdot \log(n/D) \right )} \text{ and }  \binom{\floor{n/2}}{\floor{D/2}} \cdot 2^{O \!\left (\sqrt{D\log n} \cdot \log(n/D) \right )},\]
where the upper bound is a randomized construction. From a labeling perspective our bounds are the first to give labels for any value of $D$ that are at most $o(n)$ bits longer than the shortest possible labels.
An asymptotic enumeration of the number of $D$-regular graphs due to McKay and Wormald \cite{mckay1990asymptotic,mckay1991asymptotic}
combined with Stirling's formula gives a stronger lower bound of 
$\binom{\floor{n/2}}{\floor{D/2}} \cdot n^{-O(1)}$ whenever $D = o\left(\sqrt{n}\right)$ or 
$D > \frac{cn}{\ln n}$ for a constant $c > \frac{2}{3}$.

\subsection{Acyclic graphs and cycle graphs}
On our way to understand the family
$\mathcal{G}_2$ better, we first examine two other basic families of graphs. 

For the family $\mathcal{AC}$ of acyclic graphs on $n$
vertices with maximum degree $2$
we show an upper bound matching exactly the lower bound in~\cite{Esperet2008}, which makes us conclude $g_v(\mathcal{AC})=\floor{3/2n}$.
This lower bound is not explicitly stated in~\cite{Esperet2008}, but follows directly from the construction of the lower bound for $g_v(\mathcal{G}_2)$.

We also study the family $\mathcal{C}_n$ of graphs consisting of one cycle of length $\leq n$ (and no other edges or vertices). For this family we show $n+ \Omega(\log \log n)\leq g_v(\mathcal{C}_n) \leq n+\log n + O(1)$. 

\subsection{Oblivious decoding}
From a labeling perspective one can assume all labels have the same length~\cite{AlstrupKTZ14}. Hence, if the decoder does not know what $n$ is,
it will always be able to compute $n$ approximately.
However, we show that the label size in an optimum labeling scheme
can be smaller if the decoder knows $n$ precisely.
To be more specific, let $\mathcal{F}_n$ be a family of graphs for each $n=1,2,\ldots$.
We show that there is a labeling scheme of $\mathcal F_n$ using $f(n)$ labels
that enables a decoder not aware of $n$ to answer adjacency queries
iff there is a family of graphs $G_1,G_2,\ldots$ such that $G_n$ is an induced universal
graph for $\mathcal F_n$, $|G_n|=f(n)$, and $G_n$ is an induced subgraph of
$G_{n+1}$ for every $n$.
Next we show that with this extra requirement to the induced universal graph we have
$n + \Omega\!\left(\sqrt[3]{n}\right) \leq g_v(\mathcal{C}_n) \leq n+ O(\sqrt{n})$.
The lower bound is true for infinitely many $n$, but for specific $n$ it might not hold.

For the other problems studied in this paper, the decoder does not need to know $n$, but the lower bounds hold even if it does. To the best of our knowledge this is the first time this relationship between labeling schemes and induced universal graphs
has been described and examples have been given where the complexities differ.

\subsection{Related results}
For the family of general, undirected graphs on $n$ vertices,
Alstrup \etal~\cite{AlstrupKTZ14} give an induced universal graph with
$O(2^{n/2})$ vertices, which matches a
lower bound by Moon \cite{moon1965minimal}. More recently Alon~\cite{Alonconstant2016} shows the existents of a construction having a better constant factor than the one in~\cite{AlstrupKTZ14}.

It follows from~\cite{adjacencytrees2015,alstruprauhe} that $g_v(\calA_k)=\theta(n^k)$ for the family $\calA_k$ of graphs with constant arboricity $k$ and $n$ vertices. Using  universal graphs constructed by Babai \etal~\cite{BCEGS82}, Bhatt \etal~\cite{BCLR89}, and Chung
\etal~\cite{CG78,CG79,CG83,CGP76}, Chung~\cite{Chung90} obtains the
best currently known bounds for e.g.~induced universal graphs for planar and
outerplanar bounded degree graphs.
Labeling schemes are being widely used and well-studied in the theory community:
Chung~\cite{Chung90} gives labels of size $\log n+O(\log \log n)$ for adjacency
labeling in trees, which was improved to $\log n + O(\log^* n)$~\cite{alstruprauhe} and in~\cite{bonichon2006short,Chung90,Fraigniaud2009randomized,fraigniaudkorman2,KMS02} to $\log n + \Theta(1)$ for various special cases of trees. Finally it was improved to $\log n + \Theta(1)$ for general trees~\cite{adjacencytrees2015}. 

Using labeling schemes, it is possible to avoid costly access to large
global tables and instead only perform local and distributed computations. Such properties are
used in applications such as XML search engines~\cite{AKM01}, network routing
and distributed algorithms~\cite{Cowen01,EilamGP03,Gavoille01,ThZw05}, dynamic
and parallel settings ~\cite{CohenKaplan2010,dynamicKormanP07}, and various
other applications~\cite{Korman2010,peleg2,SK85}.

A survey on induced universal graphs and adjacency labeling can be found in~\cite{AlstrupKTZ14}. See~\cite{gavoillepeleg} for a survey on
labeling schemes for various queries.

\subsection{Preliminaries}\label{sec:prelim}
Let $[n]=\{0, \ldots, n-1\}$, $\N_0 = \set{0,1,2,\ldots}$, $\N=\N_1=\set{1,2,\ldots}$,
and let $\log n$ refer to $\log_2 n$. For a graph $G$, let $V[G]$ be the set of vertices
and $E[G]$ be the set of edges of $G$,
and let $\abs{G}=\abs{V[G]}$ be the number of vertices. We denote the maximum degree of graph $G$ as
$\Delta(G)$. For $i \in \N$, let $P_{i}$ denote a path with $i$ vertices,
and for $i>2$, let $C_i$ denote a simple cycle with $i$ vertices.

Let $G$ and $U$ be two graphs and let $\lambda\colon V[G] \to V[U]$ be an injective function.
If $\lambda$ has the property that
$uv\in E[G]$ if and only if $\lambda(u)\lambda(v)\in E[U]$, we say that $\lambda$
is an \emph{embedding function} of $G$ into $U$. $G$ is an \emph{induced subgraph} of
$U$ if there exists an embedding function of $G$ into $U$, and in that case, we say that
$G$ is \emph{embedded} in $U$ and that $U$ \emph{embeds} $G$.
Let $\mathcal F$ be a family of graphs. $U$ is an \emph{induced universal graph} of
$\mathcal F$ if $G$ is an induced subgraph of $U$ for each $G\in\mathcal F$.

\newcommand{\F}{\mathcal{F}}
%Let $\F$ be a family of graphs and for each positive integer $n$ let $\F_n$ be the 
%graphs in $\F$ on $n$ vertices. For $f : \N \to \N$ we say that $\F$ can be labeled 
%using
%$f(n)$ labels with a \emph{size oblivious decoder} if there exists 
%decoder $d : \N \times \N \to \{0,1\}$ that satisfies the following: For
%every graph $G \in \F_n$ there is an encoding function
%$e : G \to \set{1,2,\ldots,f(n)}$ such that $u,v \in G$ are adjacent iff
%$d(e(u),e(v)) = 1$.
%We say that $\F$ can be labeled using $f(n)$ labels with a
%\emph{size aware decoder} if there exists a 
%decoder $d_n : \N \times \N \to \{0,1\}$ for each 
%positive integer $n$ that satisfy the following: For
%every graph $G \in \F_n$ there is an encoding function
%$e : G \to \set{1,2,\ldots,f(n)}$ such that $u,v \in G$ are adjacent iff
%$d_n(e(u),e(v)) = 1$.
Let $\F$ be a family of graphs and for each positive integer $n$, let $\F_n$ be the 
graphs in $\F$ on $n$ vertices.
For $f : \N_0 \to \N_0$ we say that $\F$ can be labeled using $f(n)$ labels with a
\emph{size aware decoder} if there exists a 
decoder $d_n : \N_0 \times \N_0 \to \{0,1\}$ for $n\in\N$,
that satisfies the following: For
every graph $G \in \F_n$ there is an encoding function
$e : G \to [f(n)]$ such that $u,v \in G$ are adjacent iff
$d_n(e(u),e(v)) = 1$.
We say that $\F$ can be labeled using $f(n)$ labels with a
\emph{size oblivious decoder} if there exists $d : \N_0 \times \N_0 \to \{0,1\}$ 
with the following property:
For every $n\in\N$ and every $G\in\F_n$, there is an encoding function
$e : G \to [f(n)]$ such that $u,v \in G$ are adjacent iff
$d(e(u),e(v)) = 1$.
The following theorem explains the relation between 
size aware and size oblivious decoders and induced universal graphs.
\begin{proposition}\label[proposition]{pro:aware}
Given a family of graphs $\F$ and a function $f : \N \to \N$ the following holds.

1) $\F$ can be labeled using $f(n)$ labels with a size aware decoder iff there
exists a family of graphs $G_1, G_2, \ldots$ such that $G_n$ is an induced universal
graph for $\F_n$ and $\abs{G_n} = f(n)$ for every $n$.

2) $\F$ can be labeled using $f(n)$ labels with a size oblivious decoder iff there
exists a family of graphs $G_1, G_2, \ldots$ such that $G_n$ is an induced universal
graph for $\F_n$, $\abs{G_n} = f(n)$, and $G_n$ is an induced subgraph of
$G_{n+1}$ for every $n$.
\end{proposition}
\begin{proof}
1) First assume that $\F$ can be labeled using $f(n)$ labels with a size aware
decoder. Let $d_n : \N_0 \times \N_0 \to \{0,1\}$ be the corresponding family
of decoders. Let $G_n$
be defined on the vertex set $[f(n)]$ such that $i \neq j$ are adjacent
iff $d_n(i,j) = 1$. Obviously, $G_n$ contains exactly $f(n)$ nodes.
Furthermore $G_n$ is an induced universal graph of $\F_n$.
This shows the first direction.

For the other direction let such a family $G_1, G_2, \ldots$ be given. Wlog assume
that the vertex set of $G_n$ is $[f(n)]$. We let
$d_n : \N_0 \times \N_0 \to \{0,1\}$ be
defined by $d_n(u,v) = 1$ iff $u,v < f(n)$ and $u,v$ are adjacent in $G_n$. For
$G \in \F_n$ let $\lambda : G \to G_n$ be an embedding function of $G$ in $G_n$.
For a node $u \in V[G]$ we assign it the label $\lambda(u)$. Then for nodes
$u,v \in V[G]$ we have $d_n(\lambda(u),\lambda(v)) = 1$ iff $\lambda(u)$ and
$\lambda(v)$ are adjacent. This happens iff $u$ and $v$ are adjacent as desired.

2) The first direction is as in the first part of the theorem. We just
need to note that $\lambda : [f(n)] \to [f(n+1)]$ is an embedding function of
$G_n$ in $G_{n+1}$.

Now consider the other direction and let such a family $G_1, G_2, \ldots$ be given
and wlog assume that the vertex set of $G_n$ is $[f(n)]$. Furthermore wlog assume 
that $\lambda_n : G_n \to G_{n+1}$ given by $k \to k$ is an embedding function of
$G_n$ in $G_{n+1}$. Let $d : \N_0 \times \N_0 \to \{0,1\}$ be defined in the
following way. For $u,v \in \N_0$ let $d(u,v) = 1$ iff there exists $n$ such that
$f(n) > u,v$ and $u$ and $v$ are adjacent in $G_n$. We note that if one such $n$ 
exists it must hold for all $n$ with $f(n) > u,v$ that $u$ and $v$ are adjacent in
$G_n$. Hence we can assign labels in the same way as the first part of the theorem.
\end{proof}

%For an inupt graph $G$ and an induced graph $U$ let \emph{embedding function} $\gamma : V[G] \mapsto V[U]$ s.t if $(u,v) \in E[G]$ then $(\lambda(u),\lambda(v)) \in E[U]$ and if $(u,v) \not\in E[G]$ then $(\lambda(u),\lambda(v)) \not\in E[U]$. We will say a subgraph being ``embedded'' in an induced universal graph, which denote the existence of such an function $\lambda$. 

% !TEX root = SODAmain.tex
\section{General $D$}

\newcommand{\G}{\mathcal{G}}

In this section we present two upper bounds on $g_v(\G_D)$, the number of nodes
in the smallest induced universal graph for graphs on $n$ nodes with bounded
degree $D$. In \Cref{thmDetUpper} we give a deterministic construction of an
induced universal graph for $\G_D$ that relies on the induced universal graph constructed
in \Cref{sec:max2upper}. In \Cref{thmRandUpper} we give a randomized construction
of an induced universal graph for $\G_D$ that with probability $\frac{1}{2}$ has
a small number of nodes. Combining the two results shows the existence of an
adjacency labeling scheme for $\G_D$ of size
$\log \binom{\floor{n/2}|}{\floor{D/2}} +
O\!\left(\min\set{D+\log n,\sqrt{D\log n}\log(n/d)}\right)$.

In \Cref{corLowerBoundLabelSizeBoundedDeg} and \Cref{corRandLowerBound}
we give lower bounds on $g_v(\G_D)$. These
lower bounds imply that any adjacency labeling scheme for $\G_D$ must have labels of
size at least
$\log \binom{\floor{n/2}|}{\floor{D/2}} -
O\!\left(\min\set{D,\sqrt{D\log n}\log(n/d)}\right)$,
which means that the upper bounds
are tight up to an additive term of size
$O\!\left(\min\set{D+\log n,\sqrt{D\log n}\log(n/d)}\right)$,
which is at most $O(\sqrt{n \log n})$. Previous labeling
schemes use labels that are larger by an additive term of size
$\Omega(D)$, which is $\Omega(n)$
when $D = \Omega(n)$, so this is the first adjacency labeling scheme for $\G_D$
where the dominating term is optimal.

\subsection{Upper bounds on $g_v(\G_D)$}
We show the following deterministic bound.
\begin{theorem}
	\label{thmDetUpper}
	For the family $\G_D$ of graphs with bounded degree $D$ on $n$ nodes
	\[
		g_v(\G_D)
		\le 
		2^{k+1} \cdot \frac{n^k}{(k-1)!}
		, \ \ \text{where} \ k = \ceil{D/2}
	\]
\end{theorem}
\begin{proof}
For a set $S$ we let $S^{\le k}$ denote the set of all subsets of $S$ of size
$\le k$. We note that $\abs{S^{\le k}} \le 2\frac{\abs{S}^k}{k!}$ whenever $S$
is finite.

We will show that $g_v(\G_D) \le 2\frac{(2n-1)^k}{(k-1)!}$. Fix $n,D$, let
$k = \ceil{D/2}$ and let $U_n$ be the induced universal graph
for $\G_2$ defined in \Cref{sec:max2upper}. We note that $V[U_n] = [2n-1]$.
We define the graph $G$ to have vertex set $[2n-1] \times [2n-1]^{\le k-1}$ and
such that there is an edge between $(x,A)$ and $(y,B)$ iff $x \in B$, $y \in A$
or $x$ and $y$ are adjacent in $U_n$. Since $G$ has the desired number of nodes
we proceed to show that $G$ is an induced universal graph for $\G_D$.
Let $H$ be a graph in $\G_D$.
By \Cref{Butlersplit} we know that we can decompose the edges of $H$ into $H_0$
and $H_1$ such that $\Delta(H_0) \le 2, \Delta(H_1) \le 2(k-1)$. We can find an
embedding function $f : V[H] \to V[U_n]$ of $H_0$ in $U_n$ by
the universality of $U_n$.
By the same argument as in the first part of the proof we can orient the edges
of $H_1$ such that any node has at most $k-1$ outgoing edges in $H_1$. For
$u \in V[H]$ let $S_u$ be the set of nodes $v$ such that there exists an edge
between $u$ and $v$ in $H_1$ oriented from $u$ to $v$. We see that $u$ and $v$
are adjacent iff $f(u)$ and $f(v)$ are adjacent in $U_n$ or it holds that 
$u \in S_v$ or $v \in S_u$. Therefore $\lambda : V[H] \to V[G]$ defined by
$u \to (f(u), f(S_u))$ is an embedding function of $H$ in $G$. Hence $G$ must
be an induced universal graph for $\G_D$.
\end{proof}

The intuition behind the randomized bound below is the following. Consider placing all $n$ vertices on a circle in a randomly chosen order and rename the vertices with indices
$[n]$ following the order on the circle.
Now, a vertex $v \in [n]$ remembers its neighbours in the next half of the circle,
i.e., $v$ stores all the adjacant vertices among
$\{v+1, \ldots, v+\lceil n/2 \rceil\}$ (where indices are taken modulu $n$).
If two vertices $u,v$ are adjacant,
then clearly either $u$ stores the index of $v$ or conversely,
hence an adjacancy query can be answered.
A Chernoff bound implies that vertex $v$ with high probability stores at most
$D/2 + O(\sqrt{D \log n})$ indices. It follows that there exists an order of the points
on the circle where every vertex stores that many neighbours and the theorem follows.
\begin{theorem}\footnote{In the full version we improve this to $\binom{\floor{n/2}}{\floor{D/2}}
		\cdot 
		2^{O \!\left (\sqrt{D\log D} \cdot \log(n/D) \right )}$}
	\label{thmRandUpper}
	For the family $\G_D$ of graphs with bounded degree $D$ on $n \ge 2D$ nodes
	\[
		g_v(\G_D) \le 
		\binom{\floor{n/2}}{\floor{D/2}}
		\cdot 
		2^{O \!\left (\sqrt{D\log n} \cdot \log(n/D) \right )}
	\]
\end{theorem}
\begin{proof}
Fix $n,D$ and wlog assume that $n$ is odd.
For $D \le \log n$ the result follows from \Cref{thmDetUpper} so assume that
$D \ge \log n$. Let $G$ be a graph in $\G_D$, and wlog assume that $V[G] = [n]$.
Let $\pi : [n] \to [n]$ be a permutation of $[n]$ chosen uniformly at random.
For each $u \in V[G]$ let $S_u$ be the set of differences
$\pi(v) - \pi(u) \bmod n$ where $v$ is a neighbour of $u$ and 
$\pi(v) - \pi(u) \bmod n$ is at most $\floor{\frac{n}{2}}+1$. That is:
\[
	S_u = \set{(\pi(v) - \pi(u)) \bmod n \mid
		(u,v) \in E[G], \
		(\pi(v) - \pi(u)) \bmod n \in \set{1,2,\ldots,\floor{\frac{n}{2}}}}
\]
Given two nodes $u,v$ we can determine whether $u$ are adjacent from $\pi(u),\pi(v)$
and $S_u,S_v$ in the following way. If
$(\pi(u)-\pi(v)) \bmod n \le \floor{\frac{n}{2}}+1$ they are adjacent iff
$(\pi(u)-\pi(v)) \bmod n \in S_v$. Otherwise they are adjacent iff
$(\pi(v)-\pi(u)) \bmod n \in S_u$.

We note that $\E(\abs{S_u}) = \frac{\deg_G(u)}{2} \le \frac{D}{2}$.
By a standard Chernoff bound without replacement we see that
\begin{align}
	\label{eqNeighbourBound}
	\abs{S_u} \le D', 
	\ \ \text{where} \ 
	D' = \floor{\frac{D}{2} + O \!\left ( \sqrt{D\log n} \right )}
\end{align}
with probability $\ge 1 - \frac{1}{2n}$ for a given vertex $u \in V[G]$.
So with probability at least $\frac{1}{2}$ we have that \eqref{eqNeighbourBound}
holds for every $u \in V[G]$. In particular there exists $\pi$ such that
\eqref{eqNeighbourBound} holds for every $u \in V[G]$. Fix such a $\pi$.

Let $D'' = \min\set{\floor{\frac{n}{2}},D'}$. Then for any node $u$ we can encode
$\pi(u)$ and $S_u$ using at most $O(\log n) + \log \binom{\floor{n/2}}{D''}$ bits.
Hence we conclude that:
\[
	g_v(\G_D) \le \binom{\floor{n/2}}{D''} n^{O(1)}
\]
The conclusion now follows from the following estimate
\[
	\binom{\floor{n/2}}{D''} \le 
	\binom{\floor{n/2}}{\floor{D/2}}
	\cdot 
	\left ( \frac{\floor{n/2}}{\floor{D/2}} \right )^{D''-\floor{D/2}}
	\le 
	\binom{\floor{n/2}}{\floor{D/2}}
	\cdot 
	\left ( \frac{n}{D} \right )^{O\!\left(\sqrt{D\log n}\right )}
\]
\end{proof}

\subsection{Lower bounds on $g_v(\G_D)$}

We now show lower bounds on  $g_v(\G_D)$. Our first lower bound follows from counting perfect matchings. 

\begin{lemma}
	\label[lemma]{lemBipartiteBoundedDegLowerBound}
	Let $n, D$ be positive integers where $n$ is even. Let $V = [n]$. The number
	of graphs $G$ with $\Delta(G) \le D$ and vertex set $V$ is at least
	$\frac{\left((n/2)!\right)^{D}}{D^{Dn/2}}$.
\end{lemma}
\begin{proof}
Let $V_0 = [n/2], V_1 = [n] \setminus [n/2]$. Let $M_0,M_1,\ldots,M_{r-1}$ be all 
perfect matchings of $V_0$ and $V_1$ where $r = (n/2)!$. Now consider the following
family of graphs being the union of $D$ such perfect matchings:
\[
  \F = \set{G \mid V[G] = V, E[G] = M_{i_0} \cup \ldots M_{i_{D-1}},
    i_0,\ldots,i_{D-1} \in [r]}
\]
Every graph in $\F$ is the union of $D$ perfect matchings and therefore has max 
degree $\le D$. Now fix $G \in \F$ and let $M$ be a perfect matching $G$. We can
write $M = \set{(u,f(u)) \mid u \in V_0}$ for some bijective function
$f : V_0 \to V_1$. There are at most $D$ ways to choose $f(u)$ for every $u \in V_0$
since $(u,f(u))$ must be an edge of $G$. Hence there are at most $D^{n/2}$ ways to
choose a perfect matching of $G$, and $G$ can be written as a union of $D$ perfect
matchings in at most $D^{Dn/2}$ ways. Since $G$ was arbitrarily chosen this must
hold for any $G \in \F$. Since there are $r^D$ ways to choose $D$ perfect matchings
we conclude that $\F$ consists of at least $\frac{r^D}{D^{nD/2}}$ graphs as desired.
\end{proof}
As an immediate corollary of \Cref{lemBipartiteBoundedDegLowerBound} we get
a lower bound on the number of nodes in an induced universal graph, shown below in \cref{corLowerBoundLabelSizeBoundedDeg}.
\begin{corollary}
	\label[corollary]{corLowerBoundLabelSizeBoundedDeg}
	The induced universal graph for the family $\G_D$ of graphs with bounded 
	degree $D$ and $n$ nodes has at least 
	$\Omega \!\left ( \left ( \frac{n}{2eD} \right )^{D/2} \right )$ nodes.
\end{corollary}
\begin{proof}
Let $G$ be the induced universal graph for the family $\G_D$. Let $V = [n]$.
Any graph $H$ from $\G_D$ on the vertex set $V$ is uniquely defined by the embedding
function $f$ of $H$ in $G$. Since there are no more than $\abs{V[G]}^n$ ways to
choose $f$, \Cref{lemBipartiteBoundedDegLowerBound} gives that
$\abs{V[G]}^n \ge \frac{\left(\floor{n/2}!\right)^{D}}{D^{D\floor{n/2}}}$.
The result now follows from Stirling's formula.
\[
\abs{V[G]} \ge 
\left (
  \frac{\left(\floor{n/2}!\right)^{2/n}}{D^{\floor{n/2}/(n/2)}}
\right )^{D/2}
\]
We note that $\floor{n/2}/(n/2) = 1$ when $n$ is even. When $n$ is odd we have
$\floor{n/2} = \frac{n-1}{2}$. Hence $\floor{n/2}/(n/2) = 1 - \frac{1}{n}$. Since
$D \le n$ we have $D^{1-\frac{1}{n}} = \Theta(D)$.
\end{proof}

Our second lower bound comes from bounding the probability that a random graph on $n$ vertices, where each edge exists with probability around $D/n$, has max degree $D$.  
\begin{lemma}
	\label{lemRandLowerBound}
	Let $n, D$ be positive integers where $n \ge 2D$.
	Let $V = [n]$. The number
	of graphs $G$ with $\Delta(G) \le D$ and vertex set $V$ is at least
	$\binom{\binom{n}{2}}{\floor{nD/2}}
	\cdot 2^{-O \!\left (n\sqrt{D\log n} \cdot \log(n/D) \right )}$.
\end{lemma}
\begin{proof}
Fix $n, D$. For $D \le \log^2 n$ the result follows from
\Cref{lemBipartiteBoundedDegLowerBound}, so assume
that $D \ge \log^2 n$.
Let $D' = D - O\!\left(\sqrt{D\log n}\right)$ be an integer.
Let $G$ be a random $G(n,p)$ graph where $p = \frac{D'}{n-1}$ and
$V[G] = [n]$. That is, $G$
is a random graph on $n$ nodes and for every pair $u,v \in V[G]$ there is an
edge between $u$ and $v$ with probability $p$.
We say that $G$ is \emph{good} if it satisfies the following two properties:
\begin{enumerate}
	\item[1] $\Delta(G) \le D$.
	\item[2] $\abs{E[G]} \ge nD''$ where $D''$ is an even integer satisfying
	  $\frac{nD''}{2} = \frac{nD'}{2} - O(\sqrt{nD})$.
\end{enumerate}
We note that $D'' = D - O(\left(\sqrt{D\log n}\right)$.
We will argue that $G$ satisfies  Property 1 with probability at least $\frac{1}{3}$.
By a Chernoff bound the probability that $u \in V[G]$ has more than $D$
neighbours is at most $\frac{1}{3n}$ if $D'$ is chosen sufficiently small.
So with probability at least $\frac{2}{3}$ we have $\Delta(G) \le D$.
Similarly, with probability at least $\frac{2}{3}$ we have
$\abs{E[G]} \ge \frac{nD'}{2} - O(\sqrt{nD})$ if we choose the constant
in the $O$-notation large enough. So with probability at least
$\frac{1}{3}$ $G$ is good.

Let $r$ be the number of good graphs and enumerate them $G_1,G_2,\ldots,G_r$.
The probability that $G = G_i$ is $p^{\abs{E[G_i]}}(1-p)^{\binom{n}{2}-\abs{E[G_i]}}$.
Since $G_i$ is good we know that $\abs{E[G_i]} \ge \frac{nD''}{2}$. Hence the probability
is at most:
\[
  p^{nD''/2}(1-p)^{\binom{n}{2}-nD''/2}
  \le 
  \binom{\binom{n}{2}}{nD''/2}^{-1}
\]
Where the inequality follows from the binomial expansion of $(p+(1-p))^{\binom{n}{2}}$.
Hence we see that:
\[
  \frac{1}{3} \le 
  \sum_{i=1}^r \Pr(G=G_i) \le 
  r\binom{\binom{n}{2}}{\frac{nD''}{2}}^{-1}
\]
And hence there are at least $\frac{1}{3}\binom{\binom{n}{2}}{nD''/2}$ graphs
with vertex set $[n]$ and maximum degree $\le D$.
Now the result follows from the following estimate:
\[
  \binom{\binom{n}{2}}{\frac{nD''}{2}}
  \ge 
  \binom{\binom{n}{2}}{\floor{\frac{nD}{2}}}
  \left ( \frac{\binom{n}{2}}{nD''/2} \right )^{nD''-nD}
  \ge 
  \binom{\binom{n}{2}}{\floor{\frac{nD}{2}}}
  \left ( \frac{n}{D} \right )^{-O\!\left(n\sqrt{D \log n}\right)}
\]
\end{proof}
As previously we get a bound on  $g_v(\G_D)$. 
\begin{corollary}
	\label[corollary]{corRandLowerBound}
	For the family $\G_D$ of graphs with bounded degree $D$ on $n \ge 2D$ nodes
	\[
		g_v(\G_D) \ge 
		\binom{\floor{n/2}}{\floor{D/2}}
		\cdot 
		2^{-O \!\left (\sqrt{D\log n} \cdot \log(n/D) \right )}
	\]
\end{corollary}
\begin{proof}
By the same argument as for \Cref{corLowerBoundLabelSizeBoundedDeg} we get that
\Cref{lemRandLowerBound} implies:
\[
  g_v(\G_D) \ge 
  \binom{\binom{n}{2}}{\floor{nD/2}}^{1/n}
  \cdot 2^{-O \!\left (\sqrt{D\log n} \cdot \log(n/D) \right )}
  =
  \binom{\floor{n/2}}{\floor{D/2}}
  \cdot 2^{-O \!\left (\sqrt{D\log n} \cdot \log(n/D) \right )}
\]
\end{proof}

\section{Paths}\label{sec:paths}
%FIXME: ordet embedded defeineret? bruge induced istedet eller lign
 \pgfdeclarelayer{background}
      \pgfdeclarelayer{foreground}
      \pgfsetlayers{background,main,foreground}

 We show that there exists an induced universal graph with $\lfloor 3n/2 \rfloor$ vertices and at most $\lfloor 3n/2 \rfloor - 1$ edges for the family $\mathcal{AC}$ of graphs consisting of a set of paths with a total of $n$ vertices.
 %As a warmup we first show an induced universal graph with $2n$ vertices that embeds the family of paths over $n$ vertices. The induced universal graph is a long path of size $2n$. It is folklore that this graph embeds any size $n$ input, and it can be seen embedding an example input graph in \Cref{fig:warmup}. The intuition of the argument is as follows: Every disjoint path of $G$ must be mapped to the induced universal graph. For every vertex of the input graph, at most two vertices of the induced universal graph must be allocated, namely in the case that the vertex is either a $P_1$ or the last vertex in a path.
Our new induced universal graph matches the following lower bound.
\begin{theorem}[Claim 1, \cite{Esperet2008}]\label{thm:pathlower}
Every induced universal graph that embeds any acyclic graph $G$ on $n$ vertices with $\Delta(G)\leq 2$ has at least $\lfloor 3n/2 \rfloor$ vertices.
\end{theorem}
\begin{proof}
The theorem can be extracted from the proof of Claim $1$ in \cite{Esperet2008}. Consider the two graphs $G'$ and $G''$, with $G'$ consisting of $n$ $P_1$ and $G''$ consising of $\floor{n/2}$ $P_2$ plus at most one $P_1$. An induced universal subgraph must contain $n$ disjoint $P_1$ to embed $G'$, and $\floor{n/2}$ disjoint $P_2$ to embed $G''$. Each $P_2$ in the embedding of $G''$ can overlap with at most one $P_1$ from the embedding of $G'$ in the induced universal graph, hence we need at least $n + \lfloor n/2 \rfloor = \lfloor 3n/2 \rfloor$ vertices in the induced universal graph.
\end{proof}
Our induced universal graph $U^p_n$ that matches \Cref{thm:pathlower} is defined as below.

\begin{definition}\label{def:upn}
  Let
  \begin{align*}
    s(x)&:=
    \begin{cases}
      x+2&\text{ if }x\equiv2\pmod{3}
      \\
      x+1&\text{ otherwise}
    \end{cases}
  \end{align*}
For any $n \in \N$ let $U^p_n$ be the graph over vertex set $\left[ \lfloor 3n/2 \rfloor \right]$, and for all $u < v \in [\lfloor 3n/2 \rfloor]$ let there be an edge $(u,v)$ iff $v = s(u)$ (See \Cref{fig:Up_n}).
\end{definition}
\begin{figure}[h]
  \centering
    \begin{tikzpicture}[scale=0.5]
      \pgfmathtruncatemacro{\n}{6};
      \pgfmathtruncatemacro{\vmax}{2*\n};
      \pgfdeclarelayer{background}
      \pgfdeclarelayer{foreground}
      \pgfsetlayers{background,main,foreground}
      \begin{scope}[
        vertex style/.style={
          draw,
          circle,
          minimum size=3mm,
          inner sep=0pt,
          outer sep=0pt%
        },
        selected vertex style/.style={
          draw,
          circle,
          %fill=black!10,
          fill=blue!20,
          minimum size=3mm,
          inner sep=0pt,
          outer sep=0pt%
        },
        selected edge style/.style={
          rounded corners,line width=1.5mm,blue!20,cap=round%
          %% %color=black, double=white,
          %% color=black!10, double=black!10,
          %% double distance=20\pgflinewidth%
        }
      ]
        \node[vertex style] (v0) at (0,-1) {\tiny $0$};
        \node[vertex style] (v1) at (0,0) {\tiny $1$};
        \node[vertex style] (v2) at (1,0) {\tiny $2$};
        \node[vertex style] (v3) at (2,-1) {\tiny $3$};
        \node[vertex style] (v4) at (2,0) {\tiny $4$};
        \node[vertex style] (v5) at (3,0) {\tiny $5$};
        \node[vertex style] (v6) at (4,-1) {\tiny $6$};
        \node[vertex style] (v7) at (4,0) {\tiny $7$};
        \node[vertex style] (v8) at (5,0) {\tiny $8$};
        \node[vertex style] (v9) at (6,-1) {\tiny $9$};
        \node[vertex style] (v10) at (6,0) {\tiny $10$};
        \node[vertex style] (v11) at (7,0) {\tiny $11$};
        \node[vertex style] (v12) at (8,-1) {\tiny $12$};
        \node[vertex style] (v13) at (8,0) {\tiny $13$};
        \node[vertex style] (v14) at (9,0) {\tiny $14$};
        \node[vertex style] (v15) at (10,-1) {\tiny $15$};

        \draw[thick,color=blue] (v0) -- (v1);
        \draw[thick,color=blue] (v3) -- (v4);
        \draw[thick,color=blue] (v6) -- (v7);
        \draw[thick,color=blue] (v9) -- (v10);
        \draw[thick,color=blue] (v12) -- (v13);

        \draw[thick,color=blue] (v1) -- (v2);
        \draw[thick,color=red] (v2) -- (v4);
        \draw[thick,color=blue] (v4) -- (v5);
        \draw[thick,color=red] (v5) -- (v7);
        \draw[thick,color=blue] (v7) -- (v8);
        \draw[thick,color=red] (v8) -- (v10);
        \draw[thick,color=blue] (v10) -- (v11);
        \draw[thick,color=red] (v11) -- (v13);
        \draw[thick,color=blue] (v13) -- (v14);

        % \begin{pgfonlayer}{background}
        %   \draw[rounded corners=2em,line width=3em,blue!20,cap=round]
        %         (v1.center) -- (v2.west) -- (v3.center);
        % \end{pgfonlayer}

      \end{scope}
    \end{tikzpicture}
    \caption{$U^p_{11}$, with each $(u,v)$ colored \textcolor{blue}{blue}
      if $v=u+1$ and \textcolor{red}{red} if $v=u+2$, i.e., the edges come from cases $1$ and $2$ of \Cref{def:upn}, respectively.}
    \label[figure]{fig:Up_n}
\end{figure}
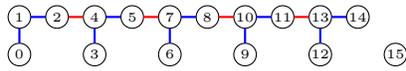
When $n$ is even $U^p_n $ is a simple caterpillar graph, where every second vertex of the main path has an extra vertex connected to it. When $n$ is odd $U^p_n$ is a caterpillar and a single isolated vertex (see \Cref{fig:Up_n}).

In the following, we show that $U^p_n$ embeds the family of acyclic graphs over $n$ vertices with maximum degree $2$, and that the corresponding decoder is size oblivious.
We proceed by showing that a few particular paths can be embedded in the graph, and finally that any graph in $\mathcal{AC}$ can be embedded using the same
embedding techniques.

\begin{definition}\label{def:bb}
Let $G$ and $U$ be graphs and let $\lambda$ be an embedding function of $G$ into $U$. We say that a vertex $v \in V[U]$ is
\emph{used} if there exists $v' \in V[G]$ such that $\lambda(v') = v$. 
A node $w\in V[U]$ is \emph{allocated} if
$w$ is used or adjacent to a used vertex.
A subset $X \subseteq V[U]$ is \emph{allocated}
when every vertex $x \in X$ is allocated.
\end{definition}

This definition allows us to argue that, given a specific graph $G\in\mathcal{AC}$ with $n$ vertices, we can allocate a part of our induced universal graph $U^p_{n}$
for embedding a part of $G$, and then the remaining unallocated part of $U^p_{n}$ is available for embedding for the remaining part of $G$. We divide $G$ into several parts and show that each part can be
embedded in $U^p_{n}$.
For each part we embed,
the number of allocated vertices divded by the number of used vertices is at most
$\lfloor 3n/2 \rfloor$. \Cref{thm:paths} implies
that these embedding strategies can be combined, i.e., applied consecutively one after the other,
implying that all of $G$ can be embedded in $U^p_{n}$.

Our argument relies heavily on allocation of blocks as defined below.
\begin{definition}\label{def:block}
Let \emph{block} $B_j$ of $U^p_n$, where $j\in[\floor{\frac{n}{2}}]$, be the vertices $\{3j, 3j+1, 3j+2 \}$.
\end{definition}

In \Cref{lem:combined}, we show how to embed single paths of any length
and simple families of paths in $U^p_{n}$ efficiently.
\Cref{fig:p2k1,fig:p3p1} shows how the cases are handled.
\begin{figure}[h]
  \centering
    \begin{tikzpicture}[scale=0.5]
      \pgfmathtruncatemacro{\n}{6};
      \pgfmathtruncatemacro{\vmax}{2*\n};
      \pgfdeclarelayer{background}
      \pgfdeclarelayer{foreground}
      \pgfsetlayers{background,main,foreground}
    \begin{scope}[
        vertex style/.style={
          draw,
          circle,
          minimum size=3mm,
          inner sep=0pt,
          outer sep=0pt%
        },
        selected vertex style/.style={
          draw,
          circle,
          %fill=black!10,
          fill=blue!20,
          minimum size=3mm,
          inner sep=0pt,
          outer sep=0pt%
        },
        selected edge style/.style={
          rounded corners,line width=1.5mm,blue!20,cap=round%
          %% %color=black, double=white,
          %% color=black!10, double=black!10,
          %% double distance=20\pgflinewidth%
        }
      ]
        \node[selected vertex style] (v0) at (0,-1) {\tiny $0$};
        \node[selected vertex style] (v1) at (0,0) {\tiny $1$};
        \node[selected vertex style] (v2) at (1,0) {\tiny $2$};
        \node[vertex style] (v3) at (2,-1) {\tiny $3$};
        \node[selected vertex style] (v4) at (2,0) {\tiny $4$};
        \node[selected vertex style] (v5) at (3,0) {\tiny $5$};
        \node[vertex style] (v6) at (4,-1) {\tiny $6$};
        \node[selected vertex style] (v7) at (4,0) {\tiny $7$};
        \node[selected vertex style] (v8) at (5,0) {\tiny $8$};
        \node[vertex style] (v9) at (6,-1) {\tiny $9$};
        \node[selected vertex style] (v10) at (6,0) {\tiny $10$};
        \node[vertex style] (v11) at (7,0) {\tiny $11$};
        \node[vertex style] (v12) at (8,-1) {\tiny $12$};
        \node[vertex style] (v13) at (8,0) {\tiny $13$};
        \node[vertex style] (v14) at (9,0) {\tiny $14$};

        \draw[thick,color=blue] (v0) -- (v1);
        \draw[thick,color=blue] (v3) -- (v4);
        \draw[thick,color=blue] (v6) -- (v7);
        \draw[thick,color=blue] (v9) -- (v10);
        \draw[thick,color=blue] (v12) -- (v13);

        \draw[thick,color=blue] (v1) -- (v2);
        \draw[thick,color=red] (v2) -- (v4);
        \draw[thick,color=blue] (v4) -- (v5);
        \draw[thick,color=red] (v5) -- (v7);
        \draw[thick,color=blue] (v7) -- (v8);
        \draw[thick,color=red] (v8) -- (v10);
        \draw[thick,color=blue] (v10) -- (v11);
        \draw[thick,color=red] (v11) -- (v13);
        \draw[thick,color=blue] (v13) -- (v14);

         \begin{pgfonlayer}{background}

          \draw[selected edge style]
                 (v0.center) -- (v1.center) -- (v2.center) -- (v4.center) -- (v5.center) -- (v7.center) -- (v8.center)
                 -- (v10.center);

          \draw[very thick,dotted,red] ($(v0.south west)+(-0.15,-0.15)$) rectangle ($(v2.north east)+(0.15,0.15)$);
          \draw[very thick,dotted,red] ($(v3.south west)+(-0.15,-0.15)$) rectangle ($(v5.north east)+(0.15,0.15)$);
          \draw[very thick,dotted,red] ($(v6.south west)+(-0.15,-0.15)$) rectangle ($(v8.north east)+(0.15,0.15)$);
          \draw[very thick,dotted,red] ($(v9.south west)+(-0.15,-0.15)$) rectangle ($(v11.north east)+(0.15,0.15)$);
         \end{pgfonlayer}
      \end{scope}
    \end{tikzpicture}
    ~~~~
     \begin{tikzpicture}[scale=0.5]
      \pgfmathtruncatemacro{\n}{6};
      \pgfmathtruncatemacro{\vmax}{2*\n};
      \pgfdeclarelayer{background}
      \pgfdeclarelayer{foreground}
      \pgfsetlayers{background,main,foreground}
      \begin{scope}[
        vertex style/.style={
          draw,
          circle,
          minimum size=3mm,
          inner sep=0pt,
          outer sep=0pt%
        },
        selected vertex style/.style={
          draw,
          circle,
          %fill=black!10,
          fill=blue!20,
          minimum size=3mm,
          inner sep=0pt,
          outer sep=0pt%
        },
        selected edge style/.style={
          rounded corners,line width=1.5mm,blue!20,cap=round%
          %% %color=black, double=white,
          %% color=black!10, double=black!10,
          %% double distance=20\pgflinewidth%
        }
      ]
        \node[selected vertex style] (v0) at (0,-1) {\tiny $0$};
        \node[selected vertex style] (v1) at (0,0) {\tiny $1$};
        \node[selected vertex style] (v2) at (1,0) {\tiny $2$};
        \node[vertex style] (v3) at (2,-1) {\tiny $3$};
        \node[selected vertex style] (v4) at (2,0) {\tiny $4$};
        \node[selected vertex style] (v5) at (3,0) {\tiny $5$};
        \node[vertex style] (v6) at (4,-1) {\tiny $6$};
        \node[selected vertex style] (v7) at (4,0) {\tiny $7$};
        \node[selected vertex style] (v8) at (5,0) {\tiny $8$};
        \node[selected vertex style] (v9) at (6,-1) {\tiny $9$};
        \node[selected vertex style] (v10) at (6,0) {\tiny $10$};
        \node[vertex style] (v11) at (7,0) {\tiny $11$};
        \node[vertex style] (v12) at (8,-1) {\tiny $12$};
        \node[vertex style] (v13) at (8,0) {\tiny $13$};
        \node[vertex style] (v14) at (9,0) {\tiny $14$};

        \draw[thick,color=blue] (v0) -- (v1);
        \draw[thick,color=blue] (v3) -- (v4);
        \draw[thick,color=blue] (v6) -- (v7);
        \draw[thick,color=blue] (v9) -- (v10);
        \draw[thick,color=blue] (v12) -- (v13);

        \draw[thick,color=blue] (v1) -- (v2);
        \draw[thick,color=red] (v2) -- (v4);
        \draw[thick,color=blue] (v4) -- (v5);
        \draw[thick,color=red] (v5) -- (v7);
        \draw[thick,color=blue] (v7) -- (v8);
        \draw[thick,color=red] (v8) -- (v10);
        \draw[thick,color=blue] (v10) -- (v11);
        \draw[thick,color=red] (v11) -- (v13);
        \draw[thick,color=blue] (v13) -- (v14);

         \begin{pgfonlayer}{background}

          \draw[selected edge style]
                 (v0.center) -- (v1.center) -- (v2.center) -- (v4.center) -- (v5.center) -- (v7.center) -- (v8.center)
                 -- (v10.center) -- (v9.center);

          \draw[very thick,dotted,red] ($(v0.south west)+(-0.15,-0.15)$) rectangle ($(v2.north east)+(0.15,0.15)$);
          \draw[very thick,dotted,red] ($(v3.south west)+(-0.15,-0.15)$) rectangle ($(v5.north east)+(0.15,0.15)$);
          \draw[very thick,dotted,red] ($(v6.south west)+(-0.15,-0.15)$) rectangle ($(v8.north east)+(0.15,0.15)$);
          \draw[very thick,dotted,red] ($(v9.south west)+(-0.15,-0.15)$) rectangle ($(v11.north east)+(0.15,0.15)$);
         \end{pgfonlayer}
      \end{scope}
    \end{tikzpicture}
    \caption{(left) $U^p_{10}$, with a $P_{8}$ embedded, and (right) $U^p_{10}$, with a $P_{9}$ embedded.  The embedded paths are shown embedded in \textcolor{blue!40}{blue}, and the allocated blocks in \textcolor{red}{red}. Both cases use at most $\lfloor 3n/2 \rfloor $ vertices.}\label[figure]{fig:p2k1}
\end{figure}
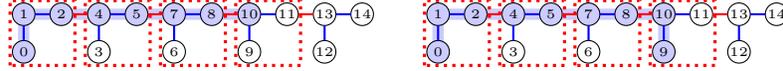

\begin{figure}[h]
  \centering
      \begin{tikzpicture}[scale=0.5]
      \pgfmathtruncatemacro{\n}{6};
      \pgfmathtruncatemacro{\vmax}{2*\n};
      \pgfdeclarelayer{background}
      \pgfdeclarelayer{foreground}
      \pgfsetlayers{background,main,foreground}
      \begin{scope}[
        vertex style/.style={
          draw,
          circle,
          minimum size=3mm,
          inner sep=0pt,
          outer sep=0pt%
        },
        selected vertex style/.style={
          draw,
          circle,
          %fill=black!10,
          fill=blue!20,
          minimum size=3mm,
          inner sep=0pt,
          outer sep=0pt%
        },
        selected edge style/.style={
          rounded corners,line width=1.5mm,blue!20,cap=round%
          %% %color=black, double=white,
          %% color=black!10, double=black!10,
          %% double distance=20\pgflinewidth%
        }
      ]
        \node[selected vertex style] (v0) at (0,-1) {\tiny $0$};
        \node[selected vertex style] (v1) at (0,0) {\tiny $1$};
        \node[selected vertex style] (v2) at (1,0) {\tiny $2$};
        \node[vertex style] (v3) at (2,-1) {\tiny $3$};
        \node[vertex style] (v4) at (2,0) {\tiny $4$};
        \node[selected vertex style] (v5) at (3,0) {\tiny $5$};
        \node[selected vertex style] (v6) at (4,-1) {\tiny $6$};
        \node[selected vertex style] (v7) at (4,0) {\tiny $7$};
        \node[vertex style] (v8) at (5,0) {\tiny $8$};
        \node[vertex style] (v9) at (6,-1) {\tiny $9$};
        \node[vertex style] (v10) at (6,0) {\tiny $10$};
        \node[vertex style] (v11) at (7,0) {\tiny $11$};
        \node[vertex style] (v12) at (8,-1) {\tiny $12$};
        \node[vertex style] (v13) at (8,0) {\tiny $13$};
        \node[vertex style] (v14) at (9,0) {\tiny $14$};

        \draw[thick,color=blue] (v0) -- (v1);
        \draw[thick,color=blue] (v3) -- (v4);
        \draw[thick,color=blue] (v6) -- (v7);
        \draw[thick,color=blue] (v9) -- (v10);
        \draw[thick,color=blue] (v12) -- (v13);

        \draw[thick,color=blue] (v1) -- (v2);
        \draw[thick,color=red] (v2) -- (v4);
        \draw[thick,color=blue] (v4) -- (v5);
        \draw[thick,color=red] (v5) -- (v7);
        \draw[thick,color=blue] (v7) -- (v8);
        \draw[thick,color=red] (v8) -- (v10);
        \draw[thick,color=blue] (v10) -- (v11);
        \draw[thick,color=red] (v11) -- (v13);
        \draw[thick,color=blue] (v13) -- (v14);

         \begin{pgfonlayer}{background}

          \draw[selected edge style]
                 (v0.center) -- (v1.center) -- (v2.center);

            \draw[selected edge style]
                 (v5.center) -- (v7.center) -- (v6.center);

          \draw[very thick,dotted,red] ($(v0.south west)+(-0.15,-0.15)$) rectangle ($(v2.north east)+(0.15,0.15)$);
          \draw[very thick,dotted,red] ($(v3.south west)+(-0.15,-0.15)$) rectangle ($(v5.north east)+(0.15,0.15)$);
          \draw[very thick,dotted,red] ($(v6.south west)+(-0.15,-0.15)$) rectangle ($(v8.north east)+(0.15,0.15)$);
          %\draw[very thick,dotted,red] ($(v9.south west)+(-0.15,-0.15)$) rectangle ($(v11.north east)+(0.15,0.15)$);
         \end{pgfonlayer}
      \end{scope}
    \end{tikzpicture}
    ~~~~
    \begin{tikzpicture}[scale=0.5]
      \pgfmathtruncatemacro{\n}{6};
      \pgfmathtruncatemacro{\vmax}{2*\n};
      \pgfdeclarelayer{background}
      \pgfdeclarelayer{foreground}
      \pgfsetlayers{background,main,foreground}
        \begin{scope}[
        vertex style/.style={
          draw,
          circle,
          minimum size=3mm,
          inner sep=0pt,
          outer sep=0pt%
        },
        selected vertex style/.style={
          draw,
          circle,
          %fill=black!10,
          fill=blue!20,
          minimum size=3mm,
          inner sep=0pt,
          outer sep=0pt%
        },
        selected edge style/.style={
          rounded corners,line width=1.5mm,blue!20,cap=round%
          %% %color=black, double=white,
          %% color=black!10, double=black!10,
          %% double distance=20\pgflinewidth%
        }
      ]
        \node[selected vertex style] (v0) at (0,-1) {\tiny $0$};
        \node[selected vertex style] (v1) at (0,0) {\tiny $1$};
        \node[selected vertex style] (v2) at (1,0) {\tiny $2$};
        \node[selected vertex style] (v3) at (2,-1) {\tiny $3$};
        \node[vertex style] (v4) at (2,0) {\tiny $4$};
        \node[selected vertex style] (v5) at (3,0) {\tiny $5$};
        \node[selected vertex style] (v6) at (4,-1) {\tiny $6$};
        \node[vertex style] (v7) at (4,0) {\tiny $7$};
        \node[selected vertex style] (v8) at (5,0) {\tiny $8$};
        \node[selected vertex style] (v9) at (6,-1) {\tiny $9$};
        \node[vertex style] (v10) at (6,0) {\tiny $10$};
        \node[selected vertex style] (v11) at (7,0) {\tiny $11$};
        \node[selected vertex style] (v12) at (8,-1) {\tiny $12$};
        \node[vertex style] (v13) at (8,0) {\tiny $13$};
        \node[vertex style] (v14) at (9,0) {\tiny $14$};

        \draw[thick,color=blue] (v0) -- (v1);
        \draw[thick,color=blue] (v3) -- (v4);
        \draw[thick,color=blue] (v6) -- (v7);
        \draw[thick,color=blue] (v9) -- (v10);
        \draw[thick,color=blue] (v12) -- (v13);

        \draw[thick,color=blue] (v1) -- (v2);
        \draw[thick,color=red] (v2) -- (v4);
        \draw[thick,color=blue] (v4) -- (v5);
        \draw[thick,color=red] (v5) -- (v7);
        \draw[thick,color=blue] (v7) -- (v8);
        \draw[thick,color=red] (v8) -- (v10);
        \draw[thick,color=blue] (v10) -- (v11);
        \draw[thick,color=red] (v11) -- (v13);
        \draw[thick,color=blue] (v13) -- (v14);

         \begin{pgfonlayer}{background}

          \draw[selected edge style]
                 (v0.center) -- (v1.center) -- (v2.center);

          \draw[selected edge style]
                 (v5.center) -- (v5.center) ;
            \draw[selected edge style]
                 (v6.center) -- (v6.center) ;
            \draw[selected edge style]
                 (v8.center) -- (v8.center) ;
            \draw[selected edge style]
                 (v9.center) -- (v9.center) ;
            \draw[selected edge style]
                 (v11.center) -- (v11.center) ;
            \draw[selected edge style]
                 (v12.center) -- (v12.center) ;

          \draw[very thick,dotted,red] ($(v0.south west)+(-0.15,-0.15)$) -- ($(v1.north west)+(-0.15,0.15)$) --  ($(v2.north east)+(0.15,0.15)$) --  ($(v2.south east)+(0.15,-0.15)$)
          --  ($(v1.south east)+(0.15,-0.15)$) --  ($(v0.south east)+(0.15,-0.15)$) --  ($(v0.south west)+(-0.15,-0.15)$) ;

          \draw[very thick,dotted,green] ($(v3.south west)+(-0.15,-0.15)$) rectangle ($(v3.north east)+(0.15,0.15)$);
           \draw[very thick,dotted,black] ($(v5.south west)+(-0.15,-1.15)$) rectangle ($(v13.north east)+(0.15,0.15)$);
          % \draw[very thick,dotted,red] ($(v6.south west)+(-0.15,-0.15)$) rectangle ($(v8.north east)+(0.15,0.15)$);
          % \draw[very thick,dotted,red] ($(v9.south west)+(-0.15,-0.15)$) rectangle ($(v11.north east)+(0.15,0.15)$);
         \end{pgfonlayer}
      \end{scope}
    \end{tikzpicture}
    \caption{(left) $U^p_{10}$ with two $P_{3}$ embedded, and (right) $U^p_{10}$ with one $P_3$ and a number of $P_1$ embedded. Embedded paths are shown in \textcolor{blue!40}{blue}, and allocated blocks in \textcolor{red}{red}. The extra allocation to handle one $P_3$ and one $P_1$ is shown in dotted \textcolor{green}{green}. The trivial extension of $6$ additional $P_1$ is shown in dotted black. All cases use at most $\lfloor 3n/2 \rfloor $ vertices.}

    \label[figure]{fig:p3p1}
\end{figure}
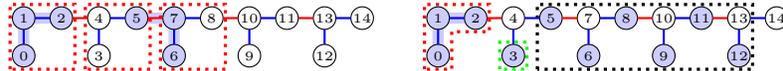

\begin{lemma}\label{lem:combined}
Let $k > 0$ unless otherwise specified. For the induced universal graph $U^p_{n}$ it holds that

\begin{enumerate}[label=(\alph*)]
  \item $U^p_{2k}$ embeds $P_{2k}$. \label{case:p2k}
  \item $U^p_{2k}$ embeds $P_{2k+1}$ when $k > 1$. \label{case:p2k1}
  \item $U^p_{k}$ embeds $k$ $P_1$. \label{case:p1}
  \item $U^p_{6}$ embeds two $P_3$. \label{case:2p3}
  \item $U^p_{3+l}$ embeds one $P_3$ and $l \geq 0$ $P_1$. \label{case:p3p1}
\end{enumerate}
\end{lemma}
\begin{proof}
We show that the cases in \Cref{lem:combined} can be embedded in our induced universal graph. The general strategy is to allocate blocks as in \Cref{def:block} and embed the input in these blocks, where at least two out of the three vertices in a block are used, yielding the desired ratio between allocated and used vertices.

Consider first case \ref{case:p2k}, where we are to embed an even path. See \Cref{fig:p2k1}. We allocate $k$ block $B_0, \ldots, B_{k-1}$ and embed the path by using three vertices in $B_0$, one vertex in $B_{k-1}$, and two vertices in $B_1, \ldots, B_{k-2}$. Say $n=2k$, then the size of the induced universal graph is thus $\lfloor 3n/2 \rfloor $. 

In the second case \ref{case:p2k1} we are to embed an odd path. See \Cref{fig:p2k1}. We use the same strategy as above, the difference being that two vertices are used in $B_{k-1}$, namely the two with the smallest label. As we embed one more vertex in the same number of blocks we are still within $\lfloor 3n/2 \rfloor $. 

For \ref{case:p1} we can embed a number $k$ of $P_1$ by allocating $\lfloor k/2 \rfloor$ blocks and embedding two $P_1$ per block. If $k$ is odd we embed the last $P_1$ by allocating the isolated vertex labelled $\floor{3k/2}-1$.

When the input is \ref{case:2p3} we embed the two $P_3$ by allocating three blocks as seen in \Cref{fig:p3p1}. As we allocate $9$ vertices and use $6$, this strategy is within the $\lfloor 3n/2 \rfloor $ bound.

In the final case \ref{case:p3p1} we allocate the first block $B_0$ and embed the $P_3$ in the block. Next we apply case \ref{case:p1} on the remaining $P_1$.  See \Cref{fig:p3p1}.  In total we then use $3 + \lfloor 3l/2 \rfloor$ vertices.
\end{proof}

We are now ready to state the main theorem of this section.
\Cref{case:main}
follows from consecutive applications of \Cref{lem:combined} and a careful order in which
we embed the vertices of some given graph $G\in\mathcal{AC}$ with $n$ vertices in
$U^p_n$.
In particular, it is crucial that the $P_3$'s and $P_1$'s in $G$
are embedded second to last and last, respectively,
since after embedding one of these two parts of $G$,
the next unallocated vertex is not necessarily the first vertex of a new block.

\begin{theorem}\label{thm:generalupperbound}\label{thm:paths}
  The graph family $U^p_0, U^p_1, \ldots$ has the property that for $n \in \N_1$

  \begin{enumerate}[label=(\alph*)]
    \item $U^p_n$ has $\lfloor 3n/2 \rfloor$ vertices and $\max\set{0, 3\floor{\frac{n}{2}}-1}$ edges.\label{case:size}
    \item $U^p_n$ is an induced subgraph of $U^p_{n+1}$.\label{case:induced}
    \item $U^p_n$ is an induced universal graph for the family of acyclic graphs $G$ with $\Delta(G)\leq2$ and $n$ vertices.\label{case:main}
  \end{enumerate}
\end{theorem}
\begin{proof}
Cases \ref{case:size} and \ref{case:induced} follow immediately from \Cref{def:upn}.

Say the input graph $G$ consists of set $S_{\mathrm{even}}$ of $P_{2k}$ for $k \geq 1$, set $S_{\mathrm{odd}}$ of $P_{2k+1}$ for $k \geq 2$, set $S_{3}$ of $P_3$, and set $S_{1}$ of $P_1$. Let the cardinality of each set be the number of paths in the set and let the total number of vertices in each set be denoted $l_{\mathrm{even}}$, $l_{\mathrm{odd}}$, $l_{3}$, and $l_1$ respectively. Observe that every acyclic graph with $\Delta(G)\leq2$ has such a partition. The proof direction is to embed the different parts of the input in the right order, such that our embedding strategy from \Cref{lem:combined} can be applied.

We start by embedding $S_{\mathrm{even}}$ by allocating the at most $3l_{\mathrm{even}}/2$ first vertices of $U^p_{n}$ by $|S_{\mathrm{even}}|$ invocations of \Cref{lem:combined} case \ref{case:p2k}. It follows that the next unallocated vertex is the first vertex of a new block, i.e., after embedding the even length paths we have spent at most $l_{\mathrm{even}}/2$ of blocks of size $3$. We allocate the next $\lfloor 3l_{\mathrm{odd}}/2 \rfloor$ vertices of $U^p_{n}$ by performing $|S_{\mathrm{odd}}|$ invocations of \Cref{lem:combined} case \ref{case:p2k1}. Again, the next unallocated vertex is the first vertex of a new block.

The next step is embedding $S_3$ and $S_1$. If $|S_3|$ is even then proceed by invoking \Cref{lem:combined} case \ref{case:2p3} $|S_3|/2$ times to embed all pairs of $P_3$. This allocation uses $3 l_3 /2$ vertices from $U^p_{n}$. We can then embed $S_1$ trivially as in \Cref{lem:combined} case \ref{case:p1}, allocating additionally $\lfloor 3 l_1 / 2 \rfloor$ vertices of $U^p_{n}$.
If $|S_3|$ is odd then allocate all the pairs using $|S_3|-1$ invocations of \Cref{lem:combined} case \ref{case:2p3}, which uses $3(l_3 - 3)/2$ vertices of $U^p_{n}$. We proceed by invoking \Cref{lem:combined} case \ref{case:p3p1} to embed the remaining $P_3$ along with $S_1$.

In total we need to allocate at most $\lfloor 3(l_{\mathrm{even}} + l_{\mathrm{odd}} + l_3 + l_1 )/2 \rfloor = \lfloor 3n/2 \rfloor$ vertices, and hence $U^p_{n}$ induces any acyclic graph $G$ on $n$ vertices where $\Delta(G)\leq2$.

\end{proof}

\section{Maximum degree $2$ upper bounds}\label{sec:max2upper}

In this section we prove upper bounds on $g_v(\mathcal{G}_2)$, and on
$g_v(\mathcal{F})$ for several special families
$\mathcal{F}\subseteq\mathcal{G}_2$.

\subsection{$2n-1$ upper bound for $g_v(\mathcal{G}_2)$}
% XXX These should be moved...
\pgfdeclarelayer{background}
\pgfdeclarelayer{foreground}
\pgfsetlayers{background,main,foreground}

\newenvironment{tikzpicture-Un}[1][1]{%
  \begin{tikzpicture}[scale=0.5]
    \pgfmathtruncatemacro{\n}{#1};
    \pgfmathtruncatemacro{\vmax}{2*\n-2};
    %% \pgfdeclarelayer{background};
    %% \pgfdeclarelayer{foreground};
    %% \pgfsetlayers{background,main,foreground};
    \begin{scope}[
        vertex style/.style={
          draw,
          circle,
          minimum size=3mm,
          inner sep=0pt,
          outer sep=0pt%
        },
        selected vertex style/.style={
          draw,
          circle,
          %fill=black!10,
          fill=blue!20,
          minimum size=3mm,
          inner sep=0pt,
          outer sep=0pt%
        },
        selected edge style/.style={
          rounded corners,line width=1.5mm,blue!20,cap=round%
          %% %color=black, double=white,
          %% color=black!10, double=black!10,
          %% double distance=20\pgflinewidth%
        }
      ]
      \foreach \i in {0,...,\vmax}{%
        \pgfmathtruncatemacro{\x}{div(\i+1,2)};
        \pgfmathtruncatemacro{\y}{mod(\i+1,2)*(2*mod(\x,2)-1)};
        \node[vertex style] (v\i) at (\x,\y) {\tiny $\i$};
      };
}{%
      \foreach \i in {0,...,\vmax}{%
        \pgfmathtruncatemacro{\nexti}{\i+1};
        \ifthenelse{\i=2 \OR \nexti>\vmax \OR \nexti=2}{%
        }{%
          \draw[thick,color=red] (v\i) -- (v\nexti);
        }

        \pgfmathtruncatemacro{\nexti}{4-mod(\i,2)+\i}
        \ifthenelse{\i=2 \OR \nexti>\vmax \OR \nexti=2}{%
        }{%
          \pgfmathtruncatemacro{\imod}{mod(\i,2)}
          \ifthenelse{\imod=1}{%
            \draw[thick,color=ForestGreen] (v\i) -- (v\nexti);
          }{%
            \draw[thick,color=blue] (v\i) -- (v\nexti);
          }
        }

      };
    \end{scope}
  \end{tikzpicture}
}%

\NewEnviron{tikzpicture-Un-induced}[2]{
  \begin{tikzpicture-Un}[#1]
    \BODY
    \begin{pgfonlayer}{background}
      \foreach \i in {#2}{%
        \node[selected vertex style] at (v\i) {\tiny $\i$};
      }
      \foreach \i in {#2}{%
        \ifthenelse{\i=2}{%
        }{%
          \foreach \j in {#2}{%
            \ifthenelse{\i<\j \AND \NOT \j=2}{%
              \pgfmathtruncatemacro{\nexti}{\i+1};
              \ifthenelse{\nexti=\j}{%
                \draw[selected edge style] (v\i.center) -- (v\j.center);
              }{}
              \pgfmathtruncatemacro{\nexti}{4-mod(\i,2)+\i}
              \ifthenelse{\nexti=\j}{%
                \draw[selected edge style] (v\i.center) -- (v\j.center);
              }{}
            }{
            }
          }
        }
      }
    \end{pgfonlayer}
  \end{tikzpicture-Un}
}%

\newcommand{\tikzUnInduced}[2][]{
  \begin{tikzpicture-Un-induced}{#2}{#1}
  \end{tikzpicture-Un-induced}
}

Here we prove that there exists an induced universal graph with $2n-1$
vertices and $4n-9$ edges for the family $\mathcal{G}_2$ of all graphs with $n$ vertices and maximum degree $2$.

\begin{definition}\label[definition]{def:Un}
  Let
  \begin{align*}
    s(x)&:=
    \begin{cases}
      x+4&\text{ if }x\equiv0\pmod{2}
      \\
      x+3&\text{ otherwise}
    \end{cases}
  \end{align*}
  and for any $n\in\mathbb{N}_0$ let $U_n$ be the graph with vertex set
  $[2n-1]$ and an edge $(u,v)$ iff $u,v\neq2$ and either
  $\abs{u-v}=1$ or $u=s(v)$ or $v=s(u)$. (See~\cref{fig:Un-example}).
\end{definition}

\begin{figure}[h!]
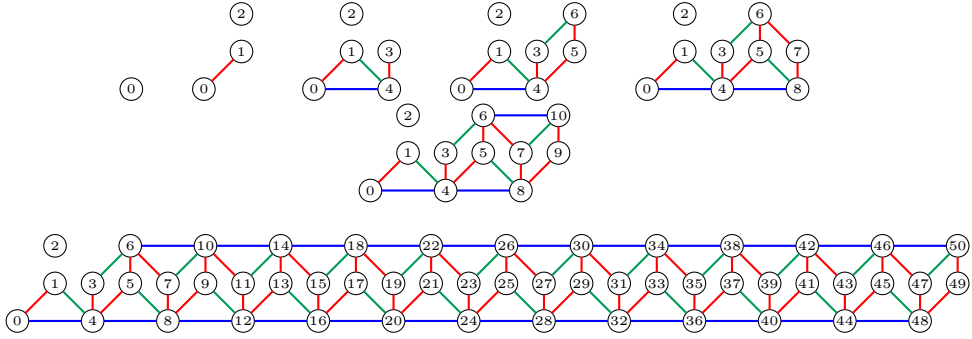

  \begin{center}
    \tikzUnInduced{1}~
    \tikzUnInduced{2}~
    \tikzUnInduced{3}~
    \tikzUnInduced{4}~
    \tikzUnInduced{5}~
    \tikzUnInduced{6}
    \\
    ~\\
    \tikzUnInduced{26}
    \caption{$U_1,\ldots,U_6$, and $U_{26}$ with each $(u,v)$ colored
      \textcolor{red}{red}/\textcolor{ForestGreen}{green}/\textcolor{blue}{blue}
      if $\abs{u-v}=1/3/4$.}
    \label{fig:Un-example}
  \end{center}
\end{figure}

%% \begin{lemma}\label[lemma]{lem:Un-neighborhood}
%%   If $X\subseteq{}V[U_n]$ and $G$ is the subgraph of $U_n$ induced by
%%   $X$, then for $v\in{}X$ the neighbors of $v$ in $G$ are:
%%   \begin{align*}
%%     N(v) &:= (X\setminus\set{2})\cap
%%     \begin{cases}
%%       \emptyset &\text{ if }v=2\\
%%       \set{v-4,v-3,v-1,v+1,v+4} &\text{ if }v\neq2 \wedge v\equiv0\pmod{2}\\
%%       \set{v-1,v+1,v+3} &\text{ otherwise}
%%     \end{cases}
%%   \end{align*}
%% \end{lemma}
%% \begin{proof}
%%   By definition, $u,v\in{}X$ are adjacent in $U_n$ iff $u,v\neq2$ and either:
%%   \begin{description}
%%   \item[$\abs{u-v}=1$:] So $u\in\set{v-1,v+1}$.
%%   \item[$u=s(v)$:] So $u=v+4$ if $v\equiv0\pmod{2}$ and $u=v+3$ otherwise.
%%   \item[$v=s(u)$:] This can only have solutions for
%%     $u\in\set{v-4,v-3}$.  If $v\equiv0\pmod{2}$ then $s(v-4)=s(v-3)=v$
%%     and both are neighbors to $v$.  On the other hand, if $v\equiv1\pmod{2}$
%%     then $s(v-4)=v-1$ and $s(v-3)=v+1$, so neither is a neighbor to
%%     $v$. \qedhere
%%   \end{description}
%% \end{proof}

\begin{theorem}
  The graph family $U_0, U_1, \ldots $ has the property that for $n\in\mathbb{N}_0$
  \begin{enumerate}[label=(\alph*)]
  \item $U_n$ has $\max\set{0,2n-1}$ vertices, and $\max\set{0,n-1,3n-5,4n-9}$ edges.
    %either $0$, $1$, $4$, or $4n-9$ edges\footnote{for $n=0/1,2,3$, and $n\geq4$ respectively}.
  \item $U_n$ is an induced subgraph of $U_{n+1}$.
  \item $U_n$ is an induced universal graph for the family of graphs with $n$ vertices and maximum degree $2$.%$\mathcal{G}_2$.
  \end{enumerate}
\end{theorem}

\begin{proof}
  If $n=0$, $U_n$ has $0$ vertices. Otherwise the number of vertices in
  $U_n$ is trivially $2n-1$ from the definition.  It is also clear
  that $U_0$ and $U_1$ each have $0$ edges, $U_2$ has $1$ edge, $U_3$ has $4$ edges,
  and $U_4$ has $7$ edges.  Finally for $n>4$, $U_{n}$ has exactly
  $4$ edges more than $U_{n-1}$, and therefore has $4(n-4)+7 = 4n-9$
  edges, as desired.

  Since the existence of an edge $(u,v)$ does not depend on $n$, the
  subgraph of $U_{n+1}$ induced by all vertices with label $\leq2n-2$ is
  exactly $U_n$.

  For the final part, consider a graph $G\in\mathcal{G}_2$.  We
  need to show that $G$ is an induced subgraph of $U_n$.  The proof is
  by induction on the number of $P_{\set{\geq3}}$ and
  $C_{\set{\geq4}}$ components in $G$.  If there are no such
  components, all components are either $P_1$, $P_2$, or $C_3$.
  Suppose therefore that $G \simeq k_1\times{}P_1 + k_2\times{}P_2 +
  k_3\times{}C_3$ for some $k_1,k_2,k_3\in\mathbb{N}_0$, and let
  $n_1=k_1$, $n_2=2k_2$, and $n_3=3k_3$ (so $n=n_1+n_2+n_3$).
  Further, assume that $n>1$ since otherwise it is trivial.
  Informally, we will show that assigning labels greedily, smallest
  label first, in the order $C_3$, $P_2$, $P_1$ is sufficient.
  Formally, let $\set{I_3,I_2,I_1}$ be the partition of
  $\set{-1,\ldots,2n-2}$ into parts of size $2n_3$, $2n_2$, and $2n_1$
  such that $i_3<i_2<i_1$ for all $(i_3,i_2,i_1)\in
  I_3\times{}I_2\times{}I_1$, and let $A_3:=I_3\setminus\set{-1,2}$,
  $A_2:=I_2\setminus\set{-1,2}$, and $A_1:=(I_1\cup\set{2})\setminus\set{-1}$.
  Then $\set{A_3,A_2,A_1}$ is a partition of $V[U_n]$.  Now let
  \begin{itemize}
  \item $V_3:=\set{i\in A_3\cond i\in\set{0,1,4}\pmod{6}}$
  \item $V_2:=\set{i\in A_2\cond i-(6n_3-1)\in\set{1,2,4,7}\pmod{8}}$
  \item $V_1:=
    \begin{cases}
      \emptyset&\text{if }n_1=0\\
      \set{2}&\text{if }n_1=1\\
      %\set{i\in A_1\cond i=2\vee (i>2(n_3+n_2)-1 \wedge i\equiv1\pmod{2})}&\text{otherwise}
      %(\set{i\in A_1\cond i\equiv1\pmod{2}}\cup\set{2})\setminus\set{2(n-n_1)-1}&\text{otherwise}
      \set{2,2(n-n_1)}\cup\set{i\in A_1\cond i\equiv1\pmod{2}\wedge i\geq2(n-n_1)+3}&\text{otherwise}
    \end{cases}$
  \end{itemize}
  Let $V=V_1\cup{}V_2\cup{}V_3$ and let $G'$ be the subgraph of $U_n$ induced by $V$.  We claim that $G\simeq{}G'$ (see~\cref{fig:V1V2V3}).
  %We claim that $G$ is the subgraph of $U_n$ induced by $V=V_1\cup{}V_2\cup{}V_3$ (see~\cref{fig:V1V2V3}).
  \begin{figure}[h!]
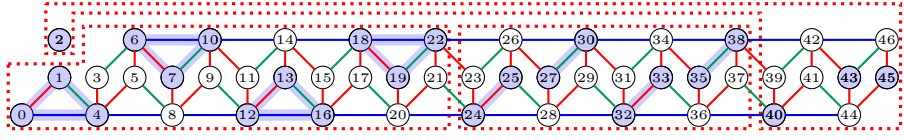

    \begin{center}
      \begin{tikzpicture-Un-induced}{24}{
          0,1,4, 6,7,10, 12,13,16, 18,19,22,
          24,25, 27,30, 32,33, 35,38,
          2,40,43,45}

        \draw[very thick,dotted,red] ($(v0.north west)+(-0.15,1.15)$)
        -- ($(v3.north west)+(-0.15,0.15)$)
        -- ($(v3.north west)+(-0.15,1.15)$)
        -- ($(v22.north east)+(+0.15,0.15)$)
        -- ($(v21.south east)+(+0.15,-1.15)$)
        -- ($(v0.south west)+(-0.15,-0.15)$)
        -- cycle;

        \draw[very thick,dotted,red] ($(v24.south west)+(-0.15,-0.15)$) rectangle  ($(v38.north east)+(0.15,0.15)$);

        \draw[very thick,dotted,red] ($(v2.north west)+(-0.15,0.75)$)
        -- ($(v46.north east)+(0.15,0.75)$)
        -- ($(v45.south east)+(0.15,-1.15)$)
        -- ($(v40.south west)+(-0.15,-0.15)$)
        -- ($(v39.north west)+(-0.15,1.5)$)
        -- ($(v2.north east)+(0.15,0.5)$)
        -- ($(v2.south east)+(0.15,-0.15)$)
        -- ($(v2.south west)+(-0.15,-0.15)$)
        -- cycle;

      \end{tikzpicture-Un-induced}
      \caption{$U_{24}$ with $4\times{}P_1+4\times{}P_2+4\times{}C_3$
        embedded.  The dotted red boxes represent $A_1$, $A_2$, and
        $A_3$ respectively.  $V_1$, $V_2$, and $V_3$ consist of the
        marked vertices in each of the dotted red boxes.}
      \label{fig:V1V2V3}
    \end{center}
  \end{figure}

  \noindent
  Now it follows from~\cref{def:Un} that for $v\in V$ the neighbors of $v$ in $G'$ are:
  \begin{align*}
    N(v) &:= (V\setminus\set{2})\cap
    \begin{cases}
      \emptyset &\text{ if }v=2\\
      \set{v-4,v-3,v-1,v+1,v+4} &\text{ if }v\neq2 \wedge v\equiv0\pmod{2}\\
      \set{v-1,v+1,v+3} &\text{ otherwise}
    \end{cases}
  \end{align*}
  To see that $G\simeq G'$, first note that $\abs{V_1}=n_1$ and $N(v_1)=\emptyset$ for all $v_1\in{}V_1$. Thus the component of $v_1$ in $G'$ is a $P_1$.
  Second, note that $\abs{V_2}=n_2$ and that each vertex $v_2\in{}V_2$
  has the form $v_2=b_2+k$ with $b_2=(6n_3-1)+8j\equiv1\pmod{2}$ for some
  $j\in[\floor{\frac{n_2}{4}}], k\in\set{1,2,4,7}$.  Now
  $N(b_2+1)=\set{b_2+2}\subseteq V_2$, $N(b_2+2)=\set{b_2+1}\subseteq V_2$,
  $N(b_2+4)=\set{b_2+7}\subseteq V_2$, and $N(b_2+7)=\set{b_2+4}\subseteq
  V_2$, so $v_2$ has exactly one neighbor in $V$, and this neighbor is
  also in $V_2$.  Thus the component of $v_2$ in $G'$ is a $P_2$.
  Third, note that $\abs{V_3}=n_3$, and that each vertex $v_3\in{}V_3$ has the
  form $v_3=b_3+k$ with $b_3=6j\equiv0\pmod{2}$ for some $j\in[\frac{n_3}{3}],
  k\in\set{0,1,4}$. Now $N(b_3+0)=\set{b_3+1,b_3+4}\subseteq V_3$,
  $N(b_3+1)=\set{b_3+0,b+4}\subseteq V_3$, and
  $N(b_3+4)=\set{b_3+0,b+1}\subseteq V_3$, so the component of $v_3$ in
  $G'$ consists of the $3$ vertices $\set{b_3,b_3+1,b_3+4}$, which form a
  $C_3$.
  Thus $G'\simeq n_1\times{}P_1 + \frac{n_2}{2}\times{}P_2 + \frac{n_3}{3}\times{}C_3 \simeq k_1\times{}P_1 + k_2\times{}P_2 +
  k_3\times{}C_3 \simeq G$.

  For the induction case, suppose $G$ has a component $X$ which is
  either a $P_k$ for some $k\geq3$, or a $C_k$ for some $k\geq4$.  In
  either case, let $I^-:=[2(n-k)-1]$ and
  $I^+:=[2n-1]\setminus{}I^-$.  Then $I^-$ induces an $U_{n-k}$
  subgraph in $U_n$, which by induction has $G-X$ as induced subgraph.
  Thus all we need to show is that we can extend this to an embedding of
  $G$ by using using only vertices from $I^+$ to embed $X$
  (see~\cref{fig:VP-VC}).
\begin{figure}[h!]
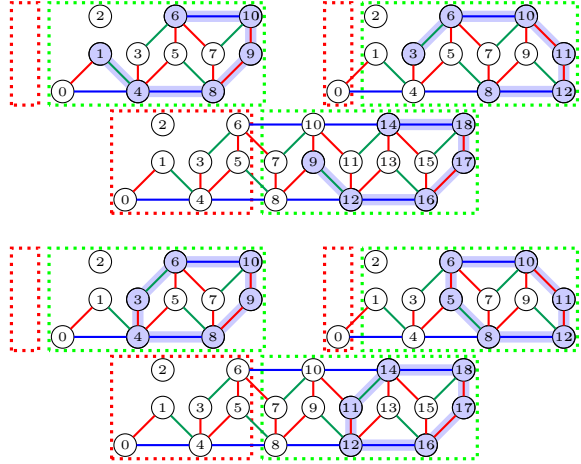
%
    \begin{center}%
      \begin{tikzpicture-Un-induced}{6}{1,4,6,8,9,10}
        \draw[very thick,dotted,red] ($(v0.south west)+(-1.15,-0.15)$) rectangle ($(v0.north east)+(-0.85,2.15)$);
        \draw[very thick,dotted,green] ($(v0.south west)+(-0.15,-0.15)$) rectangle ($(v10.north east)+(0.15,0.15)$);
      \end{tikzpicture-Un-induced}
      ~~~
      \begin{tikzpicture-Un-induced}{7}{3,6,8,10,11,12}
        \draw[very thick,dotted,red] ($(v0.south west)+(-0.15,-0.15)$) rectangle ($(v0.north east)+(0.15,2.15)$);
        \draw[very thick,dotted,green] ($(v2.north west)+(-0.15,0.15)$) rectangle ($(v12.south east)+(0.15,-0.15)$);
      \end{tikzpicture-Un-induced}
      ~~~
      \begin{tikzpicture-Un-induced}{10}{9,12,14,16,17,18}
        \draw[very thick,dotted,red] ($(v0.south west)+(-0.15,-0.15)$) rectangle ($(v6.north east)+(0.15,0.15)$);
        \draw[very thick,dotted,green] ($(v8.south west)+(-0.15,-0.15)$) rectangle ($(v18.north east)+(0.15,0.15)$);
      \end{tikzpicture-Un-induced}%
      \\
      ~
      \\
      \begin{tikzpicture-Un-induced}{6}{3,4,6,8,9,10}
        \draw[very thick,dotted,red] ($(v0.south west)+(-1.15,-0.15)$) rectangle ($(v0.north east)+(-0.85,2.15)$);
        \draw[very thick,dotted,green] ($(v0.south west)+(-0.15,-0.15)$) rectangle ($(v10.north east)+(0.15,0.15)$);
      \end{tikzpicture-Un-induced}
      ~~~
      \begin{tikzpicture-Un-induced}{7}{5,6,8,10,11,12}
        \draw[very thick,dotted,red] ($(v0.south west)+(-0.15,-0.15)$) rectangle ($(v0.north east)+(0.15,2.15)$);
        \draw[very thick,dotted,green] ($(v2.north west)+(-0.15,0.15)$) rectangle ($(v12.south east)+(0.15,-0.15)$);
      \end{tikzpicture-Un-induced}
      ~~~
      \begin{tikzpicture-Un-induced}{10}{11,12,14,16,17,18}
        \draw[very thick,dotted,red] ($(v0.south west)+(-0.15,-0.15)$) rectangle ($(v6.north east)+(0.15,0.15)$);
        \draw[very thick,dotted,green] ($(v8.south west)+(-0.15,-0.15)$) rectangle ($(v18.north east)+(0.15,0.15)$);
      \end{tikzpicture-Un-induced}%
      \\
      \caption{$P_6$ (top) and $C_6$ (bottom) embedded in $U_{6}$
        (left), $U_{7}$ (middle), and $U_{10}$ (right).  In each case
        the dotted red and green boxes represent $I^-$ and $I^+$
        respectively, and $V^P$ or $V^C$ is the marked vertices in the
        green box.  Notice that the contents of the red boxes is
        $U_0$ (left), $U_1$ (middle), and $U_4$ (right).}
      \label{fig:VP-VC}
    \end{center}
\end{figure}

  \noindent
  Now let
  \begin{itemize}
  \item $V^P:=\set{2(n-k)+1, 2n-3}\cup\set{i\in I^+\cond i\geq 2(n-k)+4 \wedge i\equiv0\pmod{2}}$
  \item $V^C:=\set{2(n-k)+3, 2n-3}\cup\set{i\in I^+\cond i\geq 2(n-k)+4 \wedge i\equiv0\pmod{2}}$
  \end{itemize}
  Now for $k\geq3$ the subgraph induced by $V^P$ in $U_n$ is $P_k$, and for all $v\in{}V^P$ all neighbors to $v$ are in $I^+$.
  Similarly, for $k\geq4$ the subgraph induced by $V^P$ in $U_n$ is $C_k$, and for all $v\in{}V^C$ all neighbors to $v$ are in $I^+$.
  Thus, either $V^P$ or $V^C$ can be used to extend the embedding of $G-X$ in $U_{n-k}$ to an embedding of $G$ in $U_n$.
\end{proof}

%% \begin{figure}[h!]
%%   \begin{center}
%%     \tikzHnInduced[1,2,3,4,5, 9,10,11,12,13,14,15]{15}
%%   \end{center}
%%   \caption{Test}
%%   \label{fig:Test}
%% \end{figure}

\subsection{$\frac{11}{6}n+\Oo(1)$ upper bound when all components are small}

%\todo{måske vi skal have helt i starten af denne sektion at vi også viser dette.}
Esperet \etal~\cite{Esperet2008} showed that $g_v(\mathcal{G}_2)\geq
11 \floor{n/6}$ by considering a specific family of graphs whose
largest component had $3$ vertices. We used the same idea in our proof of \Cref{thm:pathlower}.
%\todo{evt skrive at vi i foregående sektion brugte deres argument til at vise 3/2n?}
A natural attempt
to improve the lower bound would be to include larger components.
However, as the following shows, considering components with $4$, $5$,
or $6$ vertices is not sufficient.

\begin{figure}[h!]
  \begin{tikzpicture*}{\textwidth}
    \begin{scope}[every node/.append style={
          draw,
          circle,
          minimum size=2mm,
          inner sep=0pt,
          outer sep=0pt%
      }]

      \node (v0) at (1,1) {};
      \node (v1) at (2,1) {};
      \node (v2) at (3,1) {};
      \node (v3) at (1,2) {};
      \node (v4) at (2,2) {};
      \node (v5) at (3,2) {};
      \node (v6) at (1,3) {};
      \node (v7) at (2,3) {};
      \node (v8) at (3,3) {};
      \node (v9) at (1,4) {};
      \node (v10) at (3,4) {};
      \node (v11) at (1,5) {};
      \node (v12) at (2,5) {};
      \node (v13) at (3,5) {};
      \node (v14) at (1,6) {};
      \node (v15) at (2,6) {};
      \node (v16) at (3,6) {};
      \node (v17) at (1,7) {};
      \node (v18) at (2,7) {};
      \node (v19) at (3,7) {};

      \draw (v0) -- (v1);
      \draw (v1) -- (v2);
      \draw (v0) -- (v3);
      \draw (v1) -- (v3);
      \draw (v2) -- (v4);
      \draw (v2) -- (v5);
      \draw (v3) -- (v4);
      \draw (v3) -- (v6);
      \draw (v4) -- (v6);
      \draw (v5) -- (v7);
      \draw (v5) -- (v8);
      \draw (v6) -- (v7);
      \draw (v7) -- (v8);
      \draw (v7) -- (v9);
      \draw (v7) -- (v10);
      \draw (v9) -- (v12);
      \draw (v10) -- (v12);
      \draw (v11) -- (v12);
      \draw (v11) -- (v14);
      \draw (v12) -- (v13);
      \draw (v12) -- (v14);
      \draw (v13) -- (v15);
      \draw (v13) -- (v16);
      \draw (v14) -- (v17);
      \draw (v15) -- (v16);
      \draw (v15) -- (v17);
      \draw (v16) -- (v18);
      \draw (v16) -- (v19);
      \draw (v17) -- (v18);
      \draw (v18) -- (v19);

      \node (v20) at (1+4,1) {};
      \node (v21) at (2+4,1) {};
      \node (v22) at (3+4,1) {};
      \node (v23) at (1+4,2) {};
      \node (v24) at (2+4,2) {};
      \node (v25) at (3+4,2) {};
      \node (v26) at (1+4,3) {};
      \node (v27) at (2+4,3) {};
      \node (v28) at (3+4,3) {};
      \node (v29) at (1+4,4) {};
      \node (v30) at (3+4,4) {};
      \node (v31) at (1+4,5) {};
      \node (v32) at (2+4,5) {};
      \node (v33) at (3+4,5) {};
      \node (v34) at (1+4,6) {};
      \node (v35) at (2+4,6) {};
      \node (v36) at (3+4,6) {};
      \node (v37) at (1+4,7) {};
      \node (v38) at (2+4,7) {};
      \node (v39) at (3+4,7) {};

      \draw (v20) -- (v21);
      \draw (v21) -- (v22);
      \draw (v20) -- (v23);
      \draw (v21) -- (v23);
      \draw (v22) -- (v24);
      \draw (v22) -- (v25);
      \draw (v23) -- (v24);
      \draw (v23) -- (v26);
      \draw (v24) -- (v26);
      \draw (v25) -- (v27);
      \draw (v25) -- (v28);
      \draw (v26) -- (v27);
      \draw (v27) -- (v28);
      \draw (v27) -- (v29);
      \draw (v27) -- (v30);
      \draw (v29) -- (v32);
      \draw (v30) -- (v32);
      \draw (v31) -- (v32);
      \draw (v31) -- (v34);
      \draw (v32) -- (v33);
      \draw (v32) -- (v34);
      \draw (v33) -- (v35);
      \draw (v33) -- (v36);
      \draw (v34) -- (v37);
      \draw (v35) -- (v36);
      \draw (v35) -- (v37);
      \draw (v36) -- (v38);
      \draw (v36) -- (v39);
      \draw (v37) -- (v38);
      \draw (v38) -- (v39);

      \node (v40) at (1+14,1) {};
      \node (v41) at (2+14,1) {};
      \node (v42) at (3+14,1) {};
      \node (v43) at (1+14,2) {};
      \node (v44) at (2+14,2) {};
      \node (v45) at (3+14,2) {};
      \node (v46) at (1+14,3) {};
      \node (v47) at (2+14,3) {};
      \node (v48) at (3+14,3) {};
      \node (v49) at (1+14,4) {};
      \node (v50) at (3+14,4) {};
      \node (v51) at (1+14,5) {};
      \node (v52) at (2+14,5) {};
      \node (v53) at (3+14,5) {};
      \node (v54) at (1+14,6) {};
      \node (v55) at (2+14,6) {};
      \node (v56) at (3+14,6) {};
      \node (v57) at (1+14,7) {};
      \node (v58) at (2+14,7) {};
      \node (v59) at (3+14,7) {};

      \draw (v40) -- (v41);
      \draw (v41) -- (v42);
      \draw (v40) -- (v43);
      \draw (v41) -- (v43);
      \draw (v42) -- (v44);
      \draw (v42) -- (v45);
      \draw (v43) -- (v44);
      \draw (v43) -- (v46);
      \draw (v44) -- (v46);
      \draw (v45) -- (v47);
      \draw (v45) -- (v48);
      \draw (v46) -- (v47);
      \draw (v47) -- (v48);
      \draw (v47) -- (v49);
      \draw (v47) -- (v50);
      \draw (v49) -- (v52);
      \draw (v50) -- (v52);
      \draw (v51) -- (v52);
      \draw (v51) -- (v54);
      \draw (v52) -- (v53);
      \draw (v52) -- (v54);
      \draw (v53) -- (v55);
      \draw (v53) -- (v56);
      \draw (v54) -- (v57);
      \draw (v55) -- (v56);
      \draw (v55) -- (v57);
      \draw (v56) -- (v58);
      \draw (v56) -- (v59);
      \draw (v57) -- (v58);
      \draw (v58) -- (v59);

      \node (v60) at (11,1) {};
      \node (v61) at (11,2) {};
      \node (v62) at (10,3) {};
      \node (v63) at (11,3) {};
      \node (v64) at (12,3) {};
      \node (v65) at (10,4) {};
      \node (v66) at (12,4) {};
      \node (v67) at ( 9,5) {};
      \node (v68) at (13,5) {};
      \node (v69) at (10,6) {};
      \node (v70) at (12,6) {};
      \node (v71) at (10,7) {};
      \node (v72) at (11,7) {};
      \node (v73) at (12,7) {};
      \node (v74) at (11,8) {};
      \node (v75) at (11,9) {};

      \draw (v60) -- (v61);
      \draw (v61) -- (v62);
      \draw (v61) -- (v63);
      \draw (v61) -- (v64);
      \draw (v62) -- (v65);
      \draw (v63) -- (v65);
      \draw (v63) -- (v66);
      \draw (v64) -- (v66);
      \draw (v65) -- (v66);
      \draw (v65) -- (v67);
      \draw (v66) -- (v68);
      \draw (v67) -- (v69);
      \draw (v68) -- (v70);
      \draw (v69) -- (v70);
      \draw (v69) -- (v71);
      \draw (v69) -- (v72);
      \draw (v70) -- (v72);
      \draw (v70) -- (v73);
      \draw (v71) -- (v74);
      \draw (v72) -- (v74);
      \draw (v73) -- (v74);
      \draw (v74) -- (v75);

      \node (v76) at (11,13) {};

      \node (v77) at ( 8,10) {};
      \node (v78) at ( 9,10) {};
      \node (v79) at ( 7,11) {};
      \node (v80) at ( 8,11) {};
      \node (v81) at (10,11) {};
      \node (v82) at ( 8,12) {};
      \node (v83) at ( 9,12) {};

      \draw (v77) -- (v78);
      \draw (v77) -- (v79);
      \draw (v78) -- (v80);
      \draw (v78) -- (v81);
      \draw (v78) -- (v83);
      \draw (v79) -- (v80);
      \draw (v79) -- (v82);
      \draw (v81) -- (v83);
      \draw (v82) -- (v83);

      \node (v84) at (22- 8,10) {};
      \node (v85) at (22- 9,10) {};
      \node (v86) at (22- 7,11) {};
      \node (v87) at (22- 8,11) {};
      \node (v88) at (22-10,11) {};
      \node (v89) at (22- 8,12) {};
      \node (v90) at (22- 9,12) {};

      \draw (v84) -- (v85);
      \draw (v84) -- (v86);
      \draw (v85) -- (v87);
      \draw (v85) -- (v88);
      \draw (v85) -- (v90);
      \draw (v86) -- (v87);
      \draw (v86) -- (v89);
      \draw (v88) -- (v90);
      \draw (v89) -- (v90);

      \draw (v75) -- (v78);
      \draw (v75) -- (v85);
      \draw (v76) -- (v83);
      \draw (v76) -- (v90);

      \node (v91) at ( 2, 9) {};
      \node (v92) at ( 6, 9) {};
      \node (v93) at ( 4,10) {};
      \node (v94) at ( 3,11) {};
      \node (v95) at ( 5,11) {};
      \node (v96) at ( 4,12) {};
      \node (v97) at ( 2,13) {};
      \node (v98) at ( 6,13) {};
      \node (v99) at ( 1,11) {};

      \draw (v91) -- (v92);
      \draw (v91) -- (v93);
      \draw (v91) -- (v97);
      \draw (v92) -- (v93);
      \draw (v93) -- (v94);
      \draw (v93) -- (v95);
      \draw (v94) -- (v96);
      \draw (v95) -- (v96);
      \draw (v96) -- (v97);
      \draw (v96) -- (v98);
      \draw (v97) -- (v98);

      \draw (v97) -- (v99);

      \node (v100) at (22- 2, 9) {};
      \node (v101) at (22- 6, 9) {};
      \node (v102) at (22- 4,10) {};
      \node (v103) at (22- 3,11) {};
      \node (v104) at (22- 5,11) {};
      \node (v105) at (22- 4,12) {};
      \node (v106) at (22- 2,13) {};
      \node (v107) at (22- 6,13) {};
      \node (v108) at (22- 1,11) {};

      \draw (v100) -- (v101);
      \draw (v100) -- (v102);
      \draw (v100) -- (v106);
      \draw (v101) -- (v102);
      \draw (v102) -- (v103);
      \draw (v102) -- (v104);
      \draw (v103) -- (v105);
      \draw (v104) -- (v105);
      \draw (v105) -- (v106);
      \draw (v105) -- (v107);
      \draw (v106) -- (v107);

      \draw (v106) -- (v108);

      \draw (v79) -- (v92);
      \draw (v79) -- (v98);
      \draw (v86) -- (v101);
      \draw (v86) -- (v107);

      \node (v109) at (19,4) {};

    \end{scope}
  \end{tikzpicture*}
  \caption{This graph on $110$ nodes embeds all graphs on $60$ nodes
    whose maximum degree is $2$ and whose largest component has size
    at most $6$.}
  \label{fig:116-small}
\end{figure}
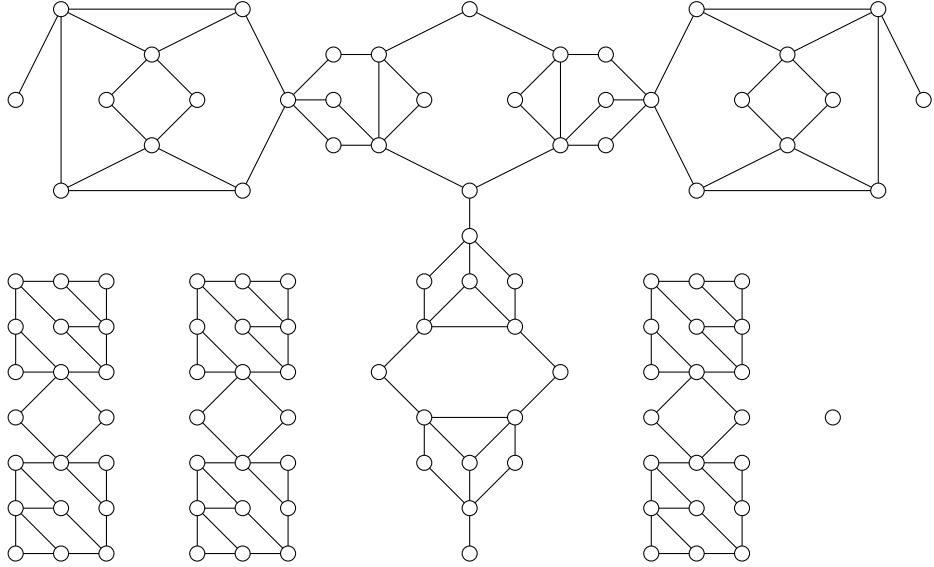

\begin{theorem}
  For any $n\in\mathbb{N}$ there exists a graph with
  $\frac{11}{6}n+\mathcal{O}(1)$ vertices, that contains as induced
  subgraphs all graphs on $n$ vertices whose maximum degree is $2$ and
  whose largest component has at most $6$ vertices.
\end{theorem}
\begin{proof}
  Let $G$ be a graph with $n$ vertices and maximum degree $2$ whose
  largest component has at most $6$ vertices.  Let $H$ be the graph on
  $110$ vertices depicted in Figure~\ref{fig:116-small}.  It is easy
  to see by inspection that this graph has each of $60\times P_1$,
  $30\times P_2$, $20\times P_3$, $20\times C_3$, $15\times P_4$,
  $15\times C_4$, $12\times P_5$, $12\times C_5$, $10\times P_6$, and
  $10\times C_6$ as induced subgraphs.  Now, any graph with the
  desired properties whose components do not include all the
  components of one of these graphs has at most
  $(60-1)+(60-2)+(60-3)+(60-3)+(60-4)+(60-4)+(60-5)+(60-5)+(60-6)+(60-6)=561$
  vertices, and can be embedded in at most $10$ copies of $H$.  The
  whole of $G$ can therefore be embedded in at most
  $c=\left\lfloor\frac{n}{60}\right\rfloor+10$ copies of $H$.  The
  total number of vertices in $c\times H$ is
  $110c\leq\frac{11}{6}n+1100$, which is
  $\frac{11}{6}n+\mathcal{O}(1)$ as desired.
\end{proof}

A more careful analysis shows that
$\left\lceil\frac{n}{60}\right\rceil$ copies of $H$ is always
sufficient to embed any $G\in\mathcal{G}_2$ with $n$ vertices, which
in particular means that the $11\floor{n/6}$ bound is achievable for
any such graph with $n$ divisible by $60$.

%% \begin{corollary}
%%   Let $\mathcal{F}\subseteq\mathcal{G}_2$.  If
%%   $g_v(\mathcal{F})\geq\frac{11}{6}n+\omega(1)$ for infinitely many
%%   $n$, then there are infinitely many graphs in $\mathcal{F}$ whose
%%   maximum component size is at least $7$.
%% \end{corollary}

We are not sure if the above construction can be extended to handle
components of size $7$ or more.  It would involve constructing a graph
with $770$ vertices.  An interesting open question is: Is there a
function $f$, such that for any $n,s\in\mathbb{N}$ the family of
graphs with $n$ vertices, maximum degree $2$, and maximum component
size $s$ has an induced universal graph with at most
$\frac{11}{6}n+f(s)$ vertices?

\subsection{$\frac{11}{6}n+\Oo(1)$ upper bound when all components are large}

Since it appears difficult to improve the lower bound by considering
only small components, the next natural thing might be to consider
just large components.  As the following upper bound shows, this is also
unlikely to succeed.

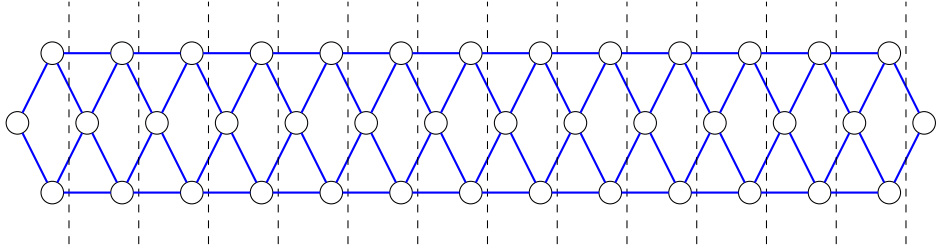
\begin{figure}[h]
  \begin{tikzpicture*}{\textwidth}

    \pgfmathtruncatemacro{\cols}{13};
    \pgfmathtruncatemacro{\p}{1};

    \begin{scope}[every node/.append style={
          draw,
          circle,
          minimum size=2mm,
          inner sep=0pt,
          outer sep=0pt%
        },
        vertex style/.style={
          draw,
          circle,
          minimum size=3mm,
          inner sep=0pt,
          outer sep=0pt%
        },
        selected vertex style/.style={
          draw,
          circle,
          %fill=black!10,
          fill=blue!20,
          minimum size=3mm,
          inner sep=0pt,
          outer sep=0pt%
        },
        selected edge style/.style={
          rounded corners,line width=1.5mm,blue!20,cap=round%
          %% %color=black, double=white,
          %% color=black!10, double=black!10,
          %% double distance=20\pgflinewidth%
        }
      ]

      \node[vertex style] (v01) at (0,2) {};

      \foreach \i in {1,...,\cols}{%
        \node[vertex style] (v\i0) at (-1+2*\i,0) {};
        \node[vertex style] (v\i1) at ( 0+2*\i,2) {};
        \node[vertex style] (v\i2) at (-1+2*\i,4) {};
        %% \pgfmathtruncatemacro{\imod}{mod(\i,\p)};
        %% \pgfmathtruncatemacro{\itop}{\i+1};
        %% \ifthenelse{\imod=1 \AND \itop<\cols}{%
        %%   \node[vertex style,color=red] (v\i3) at (2+2*\i,3.25) {};
        %% }{%
        %% }
      };

      \foreach \i in {1,...,\cols}{%
        \pgfmathtruncatemacro{\previ}{\i-1};
        \draw[thick,color=blue] (v\i0) -- (v\previ1) -- (v\i2);
        \ifthenelse{\i>1}{%
          \draw[thick,color=blue] (v\previ0) -- (v\i0);
          \draw[thick,color=blue] (v\previ2) -- (v\i2);
        }{%
        }
        \draw[thick,color=blue] (v\i0) -- (v\i1) -- (v\i2);

        \pgfmathtruncatemacro{\imod}{mod(\i,\p)};
        %% \pgfmathtruncatemacro{\itop}{\i+2};
        %% \ifthenelse{\imod=1 \AND \itop<\cols}{%
        %%   \draw[thick,color=red] (v\i2) -- (v\i3);
        %%   \draw[thick,color=red] (v\i3) -- (v\itop1);
        %%   \pgfmathtruncatemacro{\inext}{\i+3};
        %%   \draw[thick,color=red] (v\i3) -- (v\inext2);
        %%   \draw[thick,color=red] (v\i3) -- (v\i1);
        %% }{%
        %% }

        \ifthenelse{\imod=0}{%
          \draw[dashed] ($(v\i0.south east)+(0.25,-1.25)$) -- ($(v\i2.north east)+(0.25,1.25)$);
        }{%
        }
      };

    \end{scope}
  \end{tikzpicture*}
  \caption{The graph with this repeated pattern embeds in the first
    $\frac{3}{2}n-2$ vertices all graphs on $n$ vertices whose
    components are all even cycles.}
  \label{fig:32-even}
\end{figure}

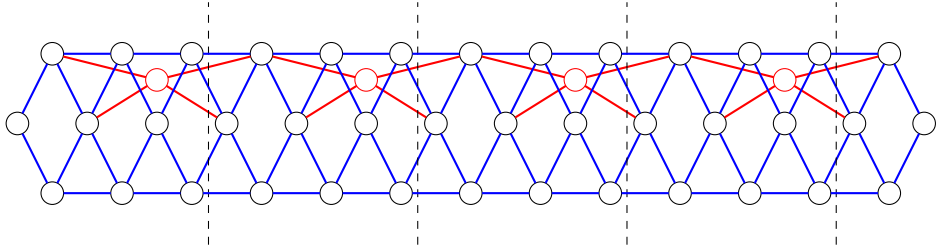
\begin{figure}[h]
  \begin{tikzpicture*}{\textwidth}

    \pgfmathtruncatemacro{\cols}{13};
    \pgfmathtruncatemacro{\p}{3};

    \begin{scope}[every node/.append style={
          draw,
          circle,
          minimum size=2mm,
          inner sep=0pt,
          outer sep=0pt%
        },
        vertex style/.style={
          draw,
          circle,
          minimum size=3mm,
          inner sep=0pt,
          outer sep=0pt%
        },
        selected vertex style/.style={
          draw,
          circle,
          %fill=black!10,
          fill=blue!20,
          minimum size=3mm,
          inner sep=0pt,
          outer sep=0pt%
        },
        selected edge style/.style={
          rounded corners,line width=1.5mm,blue!20,cap=round%
          %% %color=black, double=white,
          %% color=black!10, double=black!10,
          %% double distance=20\pgflinewidth%
        }
      ]

      \node[vertex style] (v01) at (0,2) {};

      \foreach \i in {1,...,\cols}{%
        \node[vertex style] (v\i0) at (-1+2*\i,0) {};
        \node[vertex style] (v\i1) at ( 0+2*\i,2) {};
        \node[vertex style] (v\i2) at (-1+2*\i,4) {};
        \pgfmathtruncatemacro{\imod}{mod(\i,\p)};
        \pgfmathtruncatemacro{\itop}{\i+1};
        \ifthenelse{\imod=1 \AND \itop<\cols}{%
          \node[vertex style,color=red] (v\i3) at (2+2*\i,3.25) {};
        }{%
        }
      };

      \foreach \i in {1,...,\cols}{%
        \pgfmathtruncatemacro{\previ}{\i-1};
        \draw[thick,color=blue] (v\i0) -- (v\previ1) -- (v\i2);
        \ifthenelse{\i>1}{%
          \draw[thick,color=blue] (v\previ0) -- (v\i0);
          \draw[thick,color=blue] (v\previ2) -- (v\i2);
        }{%
        }
        \draw[thick,color=blue] (v\i0) -- (v\i1) -- (v\i2);

        \pgfmathtruncatemacro{\imod}{mod(\i,\p)};
        \pgfmathtruncatemacro{\itop}{\i+2};
        \ifthenelse{\imod=1 \AND \itop<\cols}{%
          \draw[thick,color=red] (v\i2) -- (v\i3);
          \draw[thick,color=red] (v\i3) -- (v\itop1);
          \pgfmathtruncatemacro{\inext}{\i+3};
          \draw[thick,color=red] (v\i3) -- (v\inext2);
          \draw[thick,color=red] (v\i3) -- (v\i1);
        }{%
        }

        \ifthenelse{\imod=0}{%
          \draw[dashed] ($(v\i0.south east)+(0.25,-1.25)$) -- ($(v\i2.north east)+(0.25,1.25)$);
        }{%
        }
      };

    \end{scope}
  \end{tikzpicture*}
  \caption{The graph with this repeated pattern embeds in the first
    $\frac{11}{6}n+\mathcal{O}(1)$ vertices all graphs on $n$ vertices
    whose maximum degree is $2$ and whose smallest component has size
    at least $10$.}
  \label{fig:116-large}
\end{figure}

\begin{theorem}
  For any $n\in\mathbb{N}$ there exists a graph with at most
  $\frac{11}{6}n+\mathcal{O}(1)$ vertices, that contain as induced
  subgraphs all graphs on $n$ vertices whose maximum degree is $2$ and
  whose smallest component has at least $10$ vertices.
\end{theorem}
\begin{proof}
  Consider the graph family $A$ implied by the pattern in
  Figure~\ref{fig:32-even}. Observe that the leftmost
  $\frac{3}{2}n-2$ vertices of this pattern embeds all graphs on $n$
  vertices whose components are all even cycles.  In particular, each
  (even) cycle of length $s\geq10$ will be embedded in such a way that
  it uses exactly $\frac{s-4}{2}\geq3$ consecutive edges from the top
  or bottom rows.

  Now consider the modified graph family $B$ in
  Figure~\ref{fig:116-large}.  Every third edge along the top row is
  associated with one of the added red vertices.  Now consider an odd
  cycle $C$ of length $s>10$.  By splitting one vertex in the cycle,
  we can extend it to an even cycle $C'$ of length $\geq12$.  When we
  embed this $C'$, it will contain an edge associated with one of the
  red vertices.  We can therefore arrange that the split vertices are
  the endpoints of such an edge, and we can obtain an embedding of $C$
  by replacing them with the associated red vertex.

  Similarly, since the minimum component size is assumed to be $10$,
  we can turn any number of paths of total length $s$ into a cycle of
  even length by adding at most $\frac{1}{10}n+\mathcal{O}(1)$
  vertices.

  Thus any graph $G$ with $n$ vertices, maximum degree $2$ and minimum
  component size $10$ can be converted into a graph $G'$ with at most
  $n'=\frac{11}{10}n+\mathcal{O}(1)$ vertices whose components are all
  even cycles.

  The graph $G'$ can then be embedded in a graph from family $A$ with
  $\frac{3}{2}n'+\mathcal{O}(1)$ vertices.  This graph can then be
  converted to a graph in family $B$ with at most
  $\frac{10}{9}(\frac{3}{2}n'+\mathcal{O}(1))=\frac{11}{6}n+\mathcal{O}(1)$
  vertices, and the embedding of $G'$ in $A$ can be converted to an
  embedding of $G$ in $B$ as previsously described.
\end{proof}
%% \todo{måske 1-2 sætninger der siger hvad nedenstående corollary er}
%% \begin{corollary}
%%   Let $\mathcal{F}\subseteq\mathcal{G}_2$.  If
%%   $g_v(\mathcal{F})\geq\frac{11}{6}n+\omega(1)$ for infinitely many
%%   $n$, then there are infinitely many graphs in $\mathcal{F}$ whose
%%   minimum component size is at at most $9$.
%% \end{corollary}

\section{Cycle graphs}\label{sec:cycles}
We consider the family of graphs consisting of one cycle of length $\leq n$
(and no other edges or vertices).
We discuss both of the cases where the decoder is aware and oblivious of the value of $n$, as discussed in \Cref{pro:aware}. In particular we show that oblivious decoding requires a larger induced universal graph for this problem. Our new bounds leave small gaps which are interesting open problems to tighten.

First we consider the case where the decoder is aware of $n$. The new upper bound for the same setting is $n + \lg n + O(1)$ shown in \Cref{nKnown} below. The argument in \Cref{thm:awarelower} can be summarized as follows. The induced universal graph $G$ must have $n+k$ vertices for some $k$, and we wish to lower bound $k$. We use an encoding argument to show that a small encoding of the number of induced cycles on the $n+k$ vertices implies a lower bound on $k$. We show the following lower bound.

\begin{theorem}\label[theorem]{thm:awarelower}
    Let $G$ be an induced universal graph for the family of cycles of length 
    $\le n$. (I.e. the family of graphs, where each graph consists of a single
    cycle of length $\le n$.) Then $G$ has
    $n + \Omega \!\left ( \log \log n \right)$ vertices.
\end{theorem}
\begin{proof}
$G$ must contain an induced cycle of length $n$, say $(u_1,u_2,\ldots,u_n)$, and
hence it must have at least $n$ nodes, say $n+k$ for some non-negative integer $k$.
Let the remaining nodes of $G$ be ordered as $v_1,\ldots,v_k$ such that
$\deg_G(v_1) \le \ldots \le \deg_G(v_k)$.

For each $i \in \set{3,4,\ldots,n}$ let $H_i$ be an induced subgraph of $G$ that
is a cycle of length $i$. Assume that $v_j \in H_i$. Then at most $2$ of $v_j$'s
neighbours in $G$ are in $H_i$, and hence $i \le (n+k)-(\deg_G(v_j)-2)$, i.e.
$\deg_G(v_j) \le (n-i)+k+2$.

Now let $x \in \set{1,2,3,\ldots,n}$ be a parameter to be chosen later, and let $k' \in \set{1,2,\ldots,k}$ be
the largest integer such that $\deg_G(v_{k'}) \le x+k+2$.
For each $t \in \set{1,2,\ldots,x}$ we see that if $v_j \in H_{n-t}$ then
$\deg_G(v_j) \le t+k+2$ and hence $j \le k'$.

Consider $H_{n-t}$ for some $t \in \set{1,2,\ldots,x}$.
We can split the cycle into two types of paths:
\begin{itemize}
  \item[1)] A path containing a single node, one of $v_1,\ldots,v_{k'}$.
  \item[2)] A subpath of the cycle $(u_1,u_2,\ldots,u_n)$ starting and ending at
    a neighbour of one of the nodes $v_1,\ldots,v_{k'}$.
\end{itemize}
Furthermore we can split it in a way such that no two paths of type
2 are adjacent -- if they were we can merge them. Hence we can split $H_{n-t}$ in 
at most $2k'$ of such paths. The goal is now to bound the number of ways to choose
such a splitting of $H_{n-t}$ in terms of $k'$ and $\deg_G(v_i), i =1,\ldots,k'$.
Since a splitting uniquely determines $H_{n-t}$ for $t \in \set{1,2,\ldots,x}$ we
know that we can choose such a splitting in at least $x$ ways and in this way we will
get an upper bound on $x$.

Fix a cycle $H_{n-t}, t \in \set{1,\ldots,x}$ and say that it can be split into paths
in this way as $(P_1,P_2,\ldots,P_s)$ for some $s \le 2k'$. We can describe all paths
of type $1$ with a vector $w \in \set{0,1,\ldots,k'}^s$ in the following way.
For each $i$ we let $w_i = 0$ if $P_i$ is of type $2$ and otherwise we choose $w_i$
such that $P_i$ consists of the single node $v_{w_i}$. Now given $w$ we can describe
any path of type $2$, say $P_i$, by a neighbour of $w_{i-1}$, a neighbour of 
$w_{i+1}$, and a direction. (Where the indices $w_i$ are taken ${}\bmod s$.) So given
$w$, a path $P_i$ of type $2$ can be chosen in at most
$2\deg_G(v_{w_{i-1}})\deg_G(v_{w_{i+1}})$ ways. Therefore,
given $w$ the number of ways
to choose $H_{n-t}$ is no more than:
\begin{align*}
  \prod_{i=1}^s 2\deg_G(v_{w_{i-1}})\deg_G(v_{w_{i+1}}) \leq
  \prod_{i=1}^{k'} 2 (\deg_G(v_i))^2
\end{align*}
Since $w$ can be chosen in at most $\sum_{s=1}^{2k'} (k'+1)^s$ ways we conclude that:
\begin{align}
  \label{eq:boundOnX}
  x \le 
  \left ( \sum_{s=1}^{2k'} (k'+1)^s \right )
  \left ( \prod_{i=1}^{k'} 2 (\deg_G(v_i))^2 \right )
\end{align}
Now let $k_0 \in \set{2,\ldots,k}$, and suppose $x = \deg_G(v_{k_0})-k-3 > 0$. Then
$k' \le k_0-1$ and \eqref{eq:boundOnX} gives:
\begin{align}
  \deg_G(v_{k_0}) &\le 
  k+3+
  \left ( \sum_{s=1}^{2(k_0-1)} (k_0)^s \right )
  \left ( \prod_{i=1}^{k_0-1} 2 (\deg_G(v_i))^2 \right ) \notag \\
  &=
  (O(k))^{2k}
  \left ( \prod_{i=1}^{k_0-1} \deg_G(v_i) \right )^2 \label{eq:firstStepTowardsRec}
\end{align}
Furthermore, note that if $x\leq 0$ \eqref{eq:firstStepTowardsRec} trivially holds, so \eqref{eq:firstStepTowardsRec} is true for all $k_0\in\set{2,\ldots,k}$.
Let $c = c(k) = \sqrt{(O(k))^{2k}}$ be the square root of the term from 
\eqref{eq:firstStepTowardsRec}.
For $i = 1,2,\ldots,k$ let $p_i = c \prod_{j=1}^{i} \deg_G(v_j)$. Then 
\eqref{eq:firstStepTowardsRec} can be restated as $\frac{p_{k_0}}{p_{k_0-1}}\leq p_{k_0-1}^2$ or $p_{k_0} \le p_{k_0-1}^3$ for 
$k_0=2,3,\ldots,k$, and hence $p_k \le p_1^{3^{k-1}}$. On the other hand, letting
$x=n$ in \eqref{eq:boundOnX} gives $p_k^2 \cdot (O(k))^{2k} \ge n$,
and hence we have:
\begin{align*}
  n \le 
  p_1^{2 \cdot 3^{k-1}} \cdot (O(k))^{2k} =
  \left ( O(k)^k \right )^{2 \cdot 3^{k-1}+2}
\end{align*}
where we use that $\deg_G(v_1) \le k+3$ which is guaranteed by the existence of
$H_{n-1}$. Taking $\log$ twice gives $k = \Omega(\log \log n)$ as desired.
\end{proof}

Next we consider the case where the decoder is oblivious to the value $n$. The idea in the proof of \Cref{thm:oblower} is to consider a longest cycle in
each of induced universal graph $G_n$, say $a_n$. We then consider two cases:
When $a_n$ increases seldomly, and when $a_n$ increases frequently. As $a_n \ge n$
we see that if $a_n$ increases seldomly then $a_n-n$ must be large at some point,
so $\abs{G[V]}-n$ must be large. If $a_n$ increases frequently we consider the
union of all these large cycles, and prove that it must have many edges. Again
we split it into two cases. Either there exists a high degree node or all nodes
have small degree. If there exists a high degree node we use the fact that it
can only be in an induced cycle with at most two of its neighbours to conclude
that $\abs{G[V]}-n$ must be large. In the other case we look at a long induced
cycle and consider the cut between that and the rest of the graph. Since all
nodes have low degree, and there are many edges, this allows us to conclude that
$\abs{G[V]}-n$ must be large. We show the following lower bound.
\begin{theorem}\label[theorem]{thm:oblower}
    Let $G_3,G_4,G_5,\ldots$ be an infinite family of graphs, such that $G_i$
    is an induced universal graph for the family of graphs consisting of a 
    single cycle of length $\le i$. Furthermore assume that $G_i$ is an
    induced subgraph of $G_{i+1}$ for all positive integers $i$. For each $N$ 
    there exists $n \ge N$ such that $G_n$ has at least 
    $n + \Omega\!\left(\sqrt[3]{n}\right)$ vertices. In other words:
    $
        \limsup_{n \to \infty} (\abs{V[G_n]}-n)/\sqrt[3]{n}
        = \Omega(1).
    $
\end{theorem}
\begin{proof}
First we make some remarks. We can wlog assume that $G_i$ only has nodes of degree
at least $2$. We let $a_i$ denote the length of the longest induced cycle of $G_i$,
and we note that $(a_i)_{i \ge 3}$ is non-decreasing. Furthermore $a_i \ge i$.

Fix $N$ and consider $G_N,G_{N+1},\ldots,G_{2N}$. Let $k_1,k_2,\ldots$ be
defined in the following way: $k_1 = N$, and $k_{i+1}$ is the smallest index such that
$a_{k_{i+1}}$ is greater than $a_{k_i}$. We let $r$ be the highest index such that
$k_r < 2N$. We \textbf{redefine} $k_{r+1}$ such that $k_{r+1} = 2N$.
We will prove that there exists $n \in \set{N,N+1,\ldots,2N}$ such that 
\begin{align}
    \label{eqGoal}
    \abs{V[G_n]} \ge n + \Omega\!\left(\sqrt[3]{n}\right)
\end{align}
First we argue that this is true if $r < N^{2/3}$. We note that:
\[
    N = k_{r+1}-k_1 =\sum_{i = 1}^{r} k_{i+1}-k_i
\]
So there must exist $i$ such that $k_{i+1}-k_i \ge \frac{N}{r} \ge N^{1/3}$. And we
know that:
\begin{align}
    a_{k_i} = a_{k_{i+1}-1} \ge k_{i+1}-1 =
    (k_{i+1}-k_i) + k_i - 1 =
    k_i + \Omega\!\left(\sqrt[3]{k_i}\right)
\end{align}
so \eqref{eqGoal} holds when $r < N^{2/3}$ and we may from now on assume that
$r \ge N^{2/3}$.

For $i=1,2,\ldots,r$ let $C_i$ be an induced cycle in $G_{k_i}$ of length $a_{k_i}$.
Let $H_i$ be the union of $C_1,C_2,\ldots,C_i$. Then $H_i$ is an induced subgraph
of $G_{k_i}$. First assume that $C_i$ and $C_j$ are disjoint for some $i < j$. Then
$H_j$ must have at least $a_{k_i}+a_{k_j} \ge k_i+k_j$ nodes. Hence 
$\abs{V[G_{k_j}]} \ge k_j + k_i \ge k_j+N$ and \eqref{eqGoal} holds.
So assume that no two cycles $C_i,C_j$ are disjoint. Now assume that there exists a 
node in $H_i$ for some $i=1,2,\ldots,r$ with degree $>N^{1/3}$. Let $i$ be the
smallest such index and say that $\deg_{H_i}(v) > N^{1/3}$. By the minimality of $i$,
$v$ is contained in $C_i$. At most $2$ of $v$'s neighbours in $H_i$ can be contained
in $C_i$. So the number nodes in $H_i$ is at least:
\[
    a_{k_i} + \deg_{H_i}(v) - 2 >
    k_i + N^{1/3} - 2 = 
    k_i + \Omega \!\left ( \sqrt[3]{k_i} \right )
\]
Since $H_i$ is an induced subgraph of $G_{k_i}$ we see that \eqref{eqGoal} is
satisified. So from now on we can assume that $\deg_{H_i}(v) \le N^{1/3}$ for all 
$i = 1,2,\ldots,r$ and $v \in H_i$.

For a graph $G$ let $f(G)$ be defined by:
\[
    f(G) = \sum_{v \in G} \deg_G(v) - 2 = 2\abs{E[G]}-2\abs{V[G]}
\]
We note that $f(H_{i+1}) \ge f(H_i)+2$ since no two of the cycles $C_1,\ldots,C_r$ 
are disjoint. Hence $f(H_r) \ge 2(r-1)$. Let $B = V[H_r] \setminus C_r$,
and say that $b = \abs{B}$. The number of edges going between nodes in $C_r$ is 
exactly $\abs{C_r}$ since $C_r$ is a node induced cycle in $H_r$. The number of edges
going between nodes in $B$ is at most $\binom{b}{2}$. So the number of edges going
across the cut $(C_r,B)$ is at least:
\begin{align}
    \label{eqAcrossCut}
    \abs{E[H_r]} - \abs{C_r} - \binom{b}{2}
\end{align}
We see that $\frac{1}{2} f(H_r) = \abs{E[H_r]} - \abs{C_r} - b$. So 
\eqref{eqAcrossCut} is bounded from below by:
\begin{align}
    \label{eqAcrossCutTwo}
    r-1 + b - \binom{b}{2} \ge 
    r - \frac{1}{2}b^2
\end{align}
Assume that $b \le \sqrt{r}$. Then there must be at least
$\frac{1}{2}r$ edges going across 
the cut $(C_r,B)$ by \eqref{eqAcrossCutTwo}.
Since there are $b$ nodes in $B$ and each node has degree $\le N^{1/3}$ we see that
$bN^{1/3} \ge \frac{1}{2}r$. So we conclude that:
\[
    b \ge \min\set{\sqrt{r}, \frac{r}{2N^{1/3}}}
    = \Omega \!\left ( N^{1/3} \right )
\]
Since $G_{k_r}$ has $a_{k_r}+b \ge r+b$ edges we see that $n=r$ satisifies
\eqref{eqGoal}.
\end{proof}

We show constructions of induced universal graphs in the size aware and oblivious
case, respectively. 

The following lemma will turn out useful when we argue that our upper bound
constructions contain cycles of every length in a certain interval.

\begin{lemma}\label[lemma]{sumLemma}
Let $a_1\leq a_2\leq\ldots\leq a_n$ be integers such that
$a_1=1$ and $a_{i+1}\leq 2a_i$ for each
$i=1,\ldots,n-1$. Let $A=\sum_{i=1}^n a_i$. For each
$x\in\{0,\ldots,A\}$, there exists $S\subseteq\{1,\ldots,n\}$ such that
$x=\sum_{s\in S} a_s$.
\end{lemma}
\begin{proof}
The statement is trivial if $n=1$ and if $x=0$, so assume that $n>1$ and $x\geq 1$.
Let $x\in\{0,\ldots,A\}$ and choose the maximum $i\in\{1,\ldots,n\}$ such that
$a_i\leq x$.
If $i=n$, we have $x-a_n\in\{1,\ldots,A-a_n\}$ and it follows by
induction that $x-a_n$ is the sum of a subset of the numbers
$\{a_1,\ldots,a_{n-1}\}$, so the statement follows.
Consider now the case $i<n$.
Observe that
$a_{i+1}\leq 2a_i\leq a_i+2a_{i-1}\leq\ldots\leq \sum_{j=1}^i a_j+1$.
Assume for contradiction that $\sum_{j=1}^{i-1} a_j<x-a_i$.
Then $a_{i+1}-1\leq\sum_{j=1}^{i} a_j<x$ which means that
$a_{i+1}\leq x$, contradicting the maximality of $i$.
We can therefore conclude that $\sum_{j=1}^{i-1} a_j\geq x-a_i$, and it follows
by induction that $x-a_i$ is the sum of a subset of the
numbers $\{a_1,\ldots,a_{i-1}\}$, so the statement follows.
\end{proof}
The family of graphs described in \Cref{nKnown} can be thought of as a cycle $C$ of length $n$ where shortcuts are added such that also cycles of size less than $n$ are induced.
If we add new edges directly between nodes in
$C$, the graph no longer induces a cycle of size $n$.
Therefore, we have to use an extra vertex for each such shortcut.
There is a minimum length of the cycles induced in $C$ using the shortcuts,
and smaller cycles are embedded using one portion of $C$ and
a special extra vertex (labeled $u$ in \Cref{nKnownFig}).
Longer cycles are embedded in $C$ using $O(\log n)$ shortcuts.
 
\begin{theorem}\label{nKnown}
    There exists
    an infinite family $G_3,G_4,\ldots$ of graphs such that
    \begin{itemize}
    \item
    for each $n\geq 3$ and each $\ell \in\{3,\ldots,n\}$,
    the cycle of length $\ell$ is an induced subgraph of
    $G_n$, and
    \item
    $G_n$ has $n + \log n+O(1)$ nodes.
    \end{itemize}
\end{theorem}
\begin{proof}
In the following we describe the construction of $G_n$ for all
$n\geq N$ for some sufficiently large $N$ to be specified later.
For $n<N$, we only need $O(1)$ nodes to make a graph
that satisfies the requirements in the theorem.

See Figure \ref{nKnownFig}.
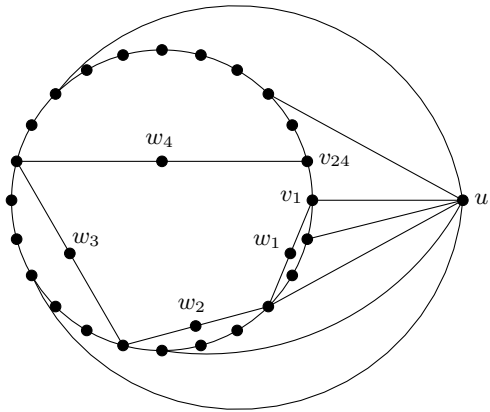
\begin{figure}[h!]
\centering
\begin{tikzpicture}[scale=2]
\clip(-1.23,-1.5) rectangle (2.3,1.5);
\draw(0.,0.) circle (1.cm);
\draw (1.,0.)-- (0.7071067811865476,-0.7071067811865468);
\draw (0.7071067811865476,-0.7071067811865468)-- (-0.25881904510252035,-0.9659258262890666);
\draw (-0.25881904510252035,-0.9659258262890666)-- (-0.9659258262890659,0.2588190451025212);
\draw (-0.9659258262890659,0.2588190451025212)-- (0.9659258262890683,0.25881904510252074);
\draw (1.,0.)-- (2.,0.);
\draw (0.965925826289068,-0.2588190451025213)-- (2.,0.);
\draw (0.7071067811865476,-0.7071067811865468)-- (2.,0.);
\draw [shift={(0.295905773966337,0.9081884520673289)}] plot[domain=4.55854279278696:5.793530385984031,variable=\t]({1.*1.9309954416491306*cos(\t r)+0.*1.9309954416491306*sin(\t r)},{0.*1.9309954416491306*cos(\t r)+1.*1.9309954416491306*sin(\t r)});
\draw [shift={(0.5045565729224794,0.10785608871552355)}] plot[domain=3.559029865141043:6.211186825080418,variable=\t]({1.*1.4993278425555803*cos(\t r)+0.*1.4993278425555803*sin(\t r)},{0.*1.4993278425555803*cos(\t r)+1.*1.4993278425555803*sin(\t r)});
\draw [shift={(0.494381579549986,-0.22861649443561366)}] plot[domain=0.15069117414675737:2.479910732145992,variable=\t]({1.*1.522876334285363*cos(\t r)+0.*1.522876334285363*sin(\t r)},{0.*1.522876334285363*cos(\t r)+1.*1.522876334285363*sin(\t r)});
\draw (0.707106781186548,0.7071067811865468)-- (2.,0.);
%\begin{scriptsize}
\draw [fill=black] (1.,0.) circle (1.0pt);
\draw[color=black] (0.862542226047113,0.022477115675780198) node {$v_1$};
\draw [fill=black] (2.,0.) circle (1.0pt);
\draw[color=black] (2.1233580965509407,0.01151349941052951) node {$u$};
\draw [fill=black] (0.9659258262890683,0.25881904510252074) circle (1.0pt);
\draw[color=black] (1.15,0.2691584816439207) node {$v_{24}$};
\draw [fill=black] (0.8660254037844389,0.5) circle (1.0pt);
\draw [fill=black] (0.707106781186548,0.7071067811865468) circle (1.0pt);
\draw [fill=black] (0.5,0.8660254037844377) circle (1.0pt);
\draw [fill=black] (0.258819045102522,0.9659258262890672) circle (1.0pt);
\draw [fill=black] (0.,1.) circle (1.0pt);
\draw [fill=black] (-0.258819045102519,0.9659258262890675) circle (1.0pt);
\draw [fill=black] (-0.5,0.866025403784438) circle (1.0pt);
\draw [fill=black] (-0.7071067811865455,0.7071067811865471) circle (1.0pt);
\draw [fill=black] (-0.8660254037844363,0.5) circle (1.0pt);
\draw [fill=black] (-0.9659258262890659,0.2588190451025212) circle (1.0pt);
%\draw[color=black] (-1.0944632773001328,0.32397656297017413) node {$v_{14}$};
\draw [fill=black] (-1.,0.) circle (1.0pt);
\draw [fill=black] (-0.9659258262890663,-0.25881904510251963) circle (1.0pt);
\draw [fill=black] (-0.8660254037844368,-0.5) circle (1.0pt);
\draw [fill=black] (-0.7071067811865461,-0.707106781186546) circle (1.0pt);
\draw [fill=black] (-0.5,-0.8660254037844372) circle (1.0pt);
\draw [fill=black] (-0.25881904510252035,-0.9659258262890666) circle (1.0pt);
%\draw[color=black] (-0.22285578421270397,-0.8162395286158974) node {$v_8$};
\draw [fill=black] (0.,-1.) circle (1.0pt);
\draw [fill=black] (0.2588190451025209,-0.9659258262890672) circle (1.0pt);
\draw [fill=black] (0.5,-0.8660254037844379) circle (1.0pt);
\draw [fill=black] (0.7071067811865476,-0.7071067811865468) circle (1.0pt);
%\draw[color=black] (0.6268244763442234,-0.586003587045633) node {$v_4$};
\draw [fill=black] (0.8660254037844384,-0.5) circle (1.0pt);
\draw [fill=black] (0.965925826289068,-0.2588190451025213) circle (1.0pt);
\draw [fill=black] (0.,0.25881904510252096) circle (1.0pt);
\draw[color=black] (-0.014547075172941137,0.4) node {$w_4$};
\draw [fill=black] (-0.6123724356957931,-0.35355339059327273) circle (1.0pt);
\draw[color=black] (-0.5024279989765963,-0.26) node {$w_3$};
\draw [fill=black] (0.2241438680420136,-0.8365163037378067) circle (1.0pt);
\draw[color=black] (0.19924344199944705,-0.7) node {$w_2$};
\draw [fill=black] (0.8535533905932737,-0.3535533905932734) circle (1.0pt);
\draw[color=black] (0.7,-0.275) node {$w_1$};
%\end{scriptsize}
\end{tikzpicture}
\caption{The graph $G_{24}$ from the proof of Theorem \ref{nKnown}.}
\label[figure]{nKnownFig}
\end{figure}
$G_n$ contains a cycle $C$ of $n$ nodes $v_1,v_2,\ldots,v_{n}$,
where the edges are $v_1v_2, v_2v_3, \ldots, v_{n}v_1$.
We define $v_{n+1}=v_1$.
There are some
additional nodes and edges as described in the following.

We define a sequence $x(1),x(2),\ldots,x(k+1)$ below and insert
shortcuts into the cycle $C$ accordingly to create shorter cycles in $G_n$.
Let $x(1)=1$. As long as $x(i)+2^{i-1}+2<n$, we define
$x(i+1)=x(i)+2^{i-1}+2$. Let $k$ be maximum such that $x(k)$ is defined
by this procedure. Let $x(k+1)=\min\{x(k)+2^{k-1}+2,n+1\}$.
Note that $x(k+1)\in\{n,n+1\}$.
$G_n$ contains $k$ \emph{shortcut nodes}
$w_1,\ldots,w_k$ and each $w_i$ is connected to $C$ by the two edges
$v_{x(i)}w_i$ and $w_iv_{x(i+1)}$.
Since $x(k)>2^{k-2}$, we have
$k\leq \log n+O(1)$.

As we shall see, the shortcut nodes ensure that $G_n$ contain cycles of all
lengths above $2k$. To make sure that $G_n$ also embeds cycles
of size at most $2k$, we define the sequence $y(2),y(3),\ldots$ as follows.
Let $y(2)=1$ and define $y(i+1)=y(i)+i-1$ for $i\geq 2$.
We see that $y(i)=2+3+\ldots+(i-2)=O(i^2)$ and hence
$y(2k)=O(\log^2 n)$.
Choose $N$ sufficiently large that $n\geq y(2k)$ for
all $n\geq N$.
Let $u$ be an extra node of $G_n$ and
connect $u$ to $C$ by the edges $uv_{y(2)},uv_{y(3)},\ldots,uv_{y(2k)}$.

It is now time to check that $G_n$ embeds cycles of all lengths $3,4,\ldots,n$.
Note that for each $\ell\in\{3,\ldots,2k\}$,
the nodes $u,v_{y(\ell-1)},v_{y(\ell-1)+1},\ldots,v_{y(\ell)}$ form a cycle
of length $\ell$.
We now show that also the cycles of lengths $\ell\in\{2k+1,\ldots,n\}$ are induced by
$G_n$. Assume first that $x(k+1)=n$.
To embed a cycle of length $\ell\geq 2k+1$, we use the nodes
$v_{x(1)},v_{x(2)},\ldots,v_{x(k+1)}$.
For each $i=1,\ldots,k$, we either use the shortcut
$v_{x(i)} w_i v_{x(i+1)}$ or the longer path
$v_{x(i)}v_{x(i)+1}\cdots v_{x(i+1)}$ on $C$.
It follows that $G_n$ induces cycles of lengths
of the form
$2k+1+\sum_{i=1}^k \delta_i\left(x(i+1)-x(i)-2\right)$ where $\delta_i\in\{0,1\}$.
Note that $x(i+1)-x(i)-2=2^{i-1}$ for $i=1,\ldots,k$.
Therefore, Lemma \ref{sumLemma} gives that the sum
$\sum_{i=1}^k \delta_i\left(x(i+1)-x(i)-2\right)$ can evaluate to any
value in the set $\{0,\ldots,n-2k-1\}$. Hence, $G_n$ induces cycles of each length
$\ell\in\{2k+1,2k+1,\ldots,n\}$.
A similar reasoning applies if $x(k+1)=n+1$.
\end{proof}

We consider the case where the decoder is oblivious of $n$. The intuition behind the construction of the graphs defined in the proof of \Cref{nUnknown} is similar to that used in the proof of \Cref{nKnown}, but we must build the graphs in a more structured manner such that $G_n$ is an induced subgraph of $G_{n+1}$. We do this by creating a long path that we carefully shortcut to create long cycles.
Each of these cycles is shortcut in a way similar to the method used in
the proof of \Cref{nKnown}. We show the following upper bound.
\begin{theorem}\label{nUnknown}
    There exists an infinite family $G_3,G_4,\ldots$
    of graphs such that
    \begin{itemize}
    \item
    for each $n\geq 3$ and each $\ell\in\{3,\ldots,n\}$,
    the cycle of length $\ell$ is an induced subgraph of
    $G_n$,
    \item
    for each $n\geq 3$, $G_n$ is an
    induced subgraph of $G_{n+1}$, and
    \item
    $G_n$ has $n + O(\sqrt{n})$ nodes.
\end{itemize}
\end{theorem}
\begin{proof}
We describe the construction of $G_n$ for indices on the form
$n=k^2+1$. If $n$ does not have that form, we define $G_n=G_m$, where
$m=\left\lceil\sqrt {n-1}\right\rceil^2+1$. Since
$m+O\!\left(\sqrt{m}\right)=n+O\!\left(\sqrt{n}\right)$, this is sufficient to
prove the theorem.

See Figure \ref{nUnknownFig}.
\begin{figure}[h!]
\begin{tikzpicture*}{\textwidth}%[scale=0.37]
%\begin{tikzpicture}[scale=0.37]
\clip(-1.6521023143510027,-9.) rectangle (38.14236031157811,3.);
\draw [shift={(1.5,2.0169425469512143)}] plot[domain=4.072932134255385:5.351845826513994,variable=\t]({1.*2.5135745936220095*cos(\t r)+0.*2.5135745936220095*sin(\t r)},{0.*2.5135745936220095*cos(\t r)+1.*2.5135745936220095*sin(\t r)});
\draw [shift={(4.,5.758772423028549)}] plot[domain=4.105301453514885:5.319476507254494,variable=\t]({1.*7.01166598036687*cos(\t r)+0.*7.01166598036687*sin(\t r)},{0.*7.01166598036687*cos(\t r)+1.*7.01166598036687*sin(\t r)});
\draw [shift={(7.5,9.556336550518099)}] plot[domain=4.046973179394067:5.377804781375312,variable=\t]({1.*12.147986181535117*cos(\t r)+0.*12.147986181535117*sin(\t r)},{0.*12.147986181535117*cos(\t r)+1.*12.147986181535117*sin(\t r)});
\draw [shift={(12.,12.819868222579434)}] plot[domain=3.960011548900102:5.464766411869277,variable=\t]({1.*17.55986962492324*cos(\t r)+0.*17.55986962492324*sin(\t r)},{0.*17.55986962492324*cos(\t r)+1.*17.55986962492324*sin(\t r)});
\draw [shift={(18.,16.61743235006898)}] plot[domain=3.887073551152627:5.537704409616752,variable=\t]({1.*24.497735771069113*cos(\t r)+0.*24.497735771069113*sin(\t r)},{0.*24.497735771069113*cos(\t r)+1.*24.497735771069113*sin(\t r)});
\draw [shift={(1.5,-1.2275487834023724)}] plot[domain=0.6858397412000377:2.4557529123897557,variable=\t]({1.*1.9382662396153538*cos(\t r)+0.*1.9382662396153538*sin(\t r)},{0.*1.9382662396153538*cos(\t r)+1.*1.9382662396153538*sin(\t r)});
\draw [shift={(5.,-1.5982115174043037)}] plot[domain=0.6741954352351088:2.4673972183546846,variable=\t]({1.*2.5601328196724027*cos(\t r)+0.*2.5601328196724027*sin(\t r)},{0.*2.5601328196724027*cos(\t r)+1.*2.5601328196724027*sin(\t r)});
\draw [shift={(10.,-3.451525187413961)}] plot[domain=0.8552715822702969:2.286321071319496,variable=\t]({1.*4.573076220592981*cos(\t r)+0.*4.573076220592981*sin(\t r)},{0.*4.573076220592981*cos(\t r)+1.*4.573076220592981*sin(\t r)});
\draw [shift={(18.,-7.575148103185446)}] plot[domain=0.9873863316245483:2.154206321965245,variable=\t]({1.*9.076500910879371*cos(\t r)+0.*9.076500910879371*sin(\t r)},{0.*9.076500910879371*cos(\t r)+1.*9.076500910879371*sin(\t r)});
\draw (0.,0.)-- (36.,0.);
\draw (1.,-6.)-- (1.5,-5.133974596215562);
\draw (2.,-6.)-- (1.5,-5.133974596215562);
\draw (1.,-6.)-- (2.,-6.);
\draw [fill=black] (0.,0.) circle (4.0pt);
\draw [fill=black] (1.,0.) circle (4.0pt);
\draw [fill=black] (2.,0.) circle (4.0pt);
\draw [fill=black] (3.,0.) circle (4.0pt);
\draw [fill=black] (4.,0.) circle (4.0pt);
\draw [fill=black] (5.,0.) circle (4.0pt);
\draw [fill=black] (6.,0.) circle (4.0pt);
\draw [fill=black] (7.,0.) circle (4.0pt);
\draw [fill=black] (8.,0.) circle (4.0pt);
\draw [fill=black] (9.,0.) circle (4.0pt);
\draw [fill=black] (10.,0.) circle (4.0pt);
\draw [fill=black] (11.,0.) circle (4.0pt);
\draw [fill=black] (12.,0.) circle (4.0pt);
\draw [fill=black] (13.,0.) circle (4.0pt);
\draw [fill=black] (14.,0.) circle (4.0pt);
\draw [fill=black] (15.,0.) circle (4.0pt);
\draw [fill=black] (16.,0.) circle (4.0pt);
\draw [fill=black] (17.,0.) circle (4.0pt);
\draw [fill=black] (18.,0.) circle (4.0pt);
\draw [fill=black] (19.,0.) circle (4.0pt);
\draw [fill=black] (20.,0.) circle (4.0pt);
\draw [fill=black] (21.,0.) circle (4.0pt);
\draw [fill=black] (22.,0.) circle (4.0pt);
\draw [fill=black] (23.,0.) circle (4.0pt);
\draw [fill=black] (24.,0.) circle (4.0pt);
\draw [fill=black] (25.,0.) circle (4.0pt);
\draw [fill=black] (26.,0.) circle (4.0pt);
\draw [fill=black] (27.,0.) circle (4.0pt);
\draw [fill=black] (28.,0.) circle (4.0pt);
\draw [fill=black] (29.,0.) circle (4.0pt);
\draw [fill=black] (30.,0.) circle (4.0pt);
\draw [fill=black] (31.,0.) circle (4.0pt);
\draw [fill=black] (32.,0.) circle (4.0pt);
\draw [fill=black] (33.,0.) circle (4.0pt);
\draw [fill=black] (34.,0.) circle (4.0pt);
\draw [fill=black] (35.,0.) circle (4.0pt);
\draw [fill=black] (36.,0.) circle (4.0pt);
\draw[color=black] (-0.8,0) node {$v_1$};
\draw[color=black] (37,0) node {$v_{36}$};
\draw [fill=black] (1.5,-0.4966320466707952) circle (4.0pt);
\draw [fill=black] (4.,-1.2528935573383215) circle (4.0pt);
\draw [fill=black] (7.5,-2.591649631017018) circle (4.0pt);
\draw [fill=black] (12.,-4.740001402343806) circle (4.0pt);
\draw [fill=black] (18.,-7.880303421000132) circle (4.0pt);
\draw [fill=black] (1.5,0.7107174562129814) circle (4.0pt);
\draw [fill=black] (5.,0.9619213022680989) circle (4.0pt);
\draw [fill=black] (10.,1.1215510331790193) circle (4.0pt);
\draw [fill=black] (18.,1.5013528076939253) circle (4.0pt);
\draw [fill=black] (1.,-6.) circle (4.0pt);
\draw [fill=black] (2.,-6.) circle (4.0pt);
\draw [fill=black] (1.5,-5.133974596215562) circle (4.0pt);
\end{tikzpicture*}
%\end{tikzpicture}
\caption{The graph $G_{37}$ from the proof of Theorem \ref{nUnknown}.
The nodes $u_2,\ldots,u_6$ are below the horizontal path $v_1v_2\cdots v_{36}$.
The nodes $w_1,\ldots,w_4$ are above.}
\label[figure]{nUnknownFig}
\end{figure}
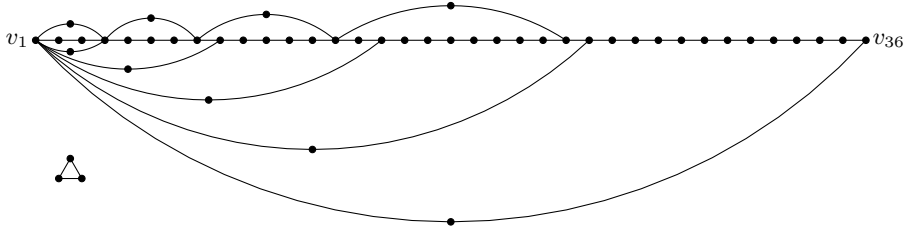
Let $n=k^2+1$ and define $G_n$ as follows. $G_n$ contains a path $P$ consisting
of the nodes
$v_1v_2\cdots v_{k^2}$.
Let $I_{\square}=\{1,4,\ldots,k^2\}$.
For each $i^2\in I_{\square}$ where $i^2>1$, there is a node $u_i$ and edges
$v_{i^2}u_i$ and $u_iv_1$. Note that the graph described so far embeds
cycles of each length $2^2+1,3^2+1,\ldots,k^2+1$.

Similar to the construction described in the proof of Theorem \ref{nKnown},
we define shortcuts in $P$ between nodes with
indices $x(1),x(2),\ldots$ as described below.
Define $x(1)=1$ and recursively $x(i+1)=x(i)+2^{i-1}+2$.
It follows that $x(i+1)=2^i+2i$.
Let $X=\{x(i)\mid i\in\mathbb N\textrm{ and }x(i)\leq k^2\}$.
Whenever $\{x(i),x(i+1)\}\subset X$ for some $i$, $G_n$ has
a node $w_i$ and edges $v_{x(i)}w_i$ and $w_iv_{x(i+1)}$.

For $p\leq k$,
let $L(p)$ be the set of lengths of cycles that %among other nodes
contain %$v_1,v_4,\ldots,v_{p^2}$, and 
$u_p$.
In the following, we see that
$\{7,8,\ldots,n\}\subset\bigcup_{p=3}^k L(p)$, and it follows that
we can obtain the required properties by adding $O(1)$ nodes to
$G_n$.

For a number $4\leq p^2\leq k^2$, we consider the maximum
$x(r)\in X$ such that $x(r)\leq p^2$, and note that $r\geq2$. We write
$p^2=x(r)+q$, where $0\leq q<x(r+1)-x(r)=2^{r-1}+2$.
Therefore $q/p^2=\frac{q}{x(r)+q}\leq
\frac{2^{r-1}+1}{2^r+2r-1}<1/2$, and
the path $v_{x(r)}v_{x(r)+1}\cdots v_{p^2}$ consists of
less than $p^2/2+1$ nodes.

Since $x(i)>2^{i-1}$ for all $i$,
there are no more than $\lfloor\log p^2\rfloor$ numbers
in $X\cap\{2,3,\ldots,p^2\}$.
We now construct a cycle in the following way. For each $i=1,2,\ldots,r-1$, we
either use the shortcut $v_{x(i)}w_iv_{x(i+1)}$ or the longer path
$v_{x(i)}v_{x(i)+1}\cdots v_{x(i+1)}$ on $P$. Then follows the path
$v_{x(r)}v_{x(r)+1}\cdots v_{p^2}u_p$ which completes the cycle.
Let $c(p)$ be the length of the smallest such cycle and let
$d(p)=p^2/2+2\lfloor\log p^2\rfloor+2$.
By the above discussion, we have $c(p)\leq d(p)$.
Furthermore, Lemma \ref{sumLemma} implies that one can construct the cycle to have
any length in the set $\{c(p),c(p)+1,\ldots,p^2+1\}$. It is seen that
$d(p)\leq (p-1)^2+1$ for $p\geq 6$. Hence, when $n\geq 5^2+1=26$,
$G_n$ embeds any cycle of length $\{c(5),c(5)+1,\ldots,k^2+1\}$.
We also have $c(3)=7$, $c(4)=10\leq 3^2+1$, $c(5)=11<4^2+1$. Therefore,
the only relevant cycles not induced by $G_n$ have length less than $7$.
It is seen that $G_5$ is lacking a cycle of length $3$ -- otherwise,
the graphs contain all the required cycles.
This is handled by adding a triangle to $G_n$.
\end{proof}

\bibliographystyle{amsplain}
\bibliography{max2bib}

\end{document}